\documentclass[10pt,journal,letterpaper]{IEEEtran}

\usepackage{setspace} 
\usepackage{enumerate}
\usepackage{xcolor}
\usepackage{cite}
\usepackage[cmex10]{amsmath}
\usepackage{amsthm}
\usepackage{graphicx}
\graphicspath{{./figs/}}
\usepackage[caption=false,font=footnotesize]{subfig}
\usepackage{bm}
\usepackage{pifont} 
\usepackage{mathrsfs}
\usepackage{amsfonts}
\usepackage{mathabx}
\usepackage{physics}
\usepackage{enumitem}
\usepackage{balance}

\newcommand{\nop}[1]{}

\newtheorem{theorem}{Theorem}
\newtheorem{corollary}{Corollary}
\newtheorem{lemma}{Lemma}

\newtheorem{remark}{Remark}

\renewcommand{\Pr}{\mathbb{P}}

\usepackage{tikz}
\usetikzlibrary{matrix,arrows,decorations.pathmorphing}

\newcommand*\circled[1]{\tikz[baseline=(char.base)]{
            \node[shape=circle,draw,inner sep=2pt] (char) {#1};}}

\begin{document}

\title{Dependence Control at Large}
\author{
Fengyou Sun
\thanks{The author is with the Department of Information Security and Communication Technology, NTNU -- Norwegian University of Science and Technology, Trondheim, Norway (Email: sunfengyou@gmail.com).}
}

\maketitle

\begin{abstract}
We study the dependence control theory, with a focus on the tail property and dependence transformability of wireless channel capacity, respectively, from the perspective of an information theoretic model of the wireless channel and from the perspective of a functional of controllable and uncontrollable random parameter processes.
We find that the light-tailed behavior is an intrinsic property of the wireless channel capacity, which is due to the passive nature of the wireless propagation environment and the power limitation in the practical systems.
We observe that the manipulation of the marginal distributions has a bias in favor of positive dependence and against negative dependence, e.g., when a parameter process bears negative dependence, the increases of the means of marginals can not leads effectively to a better system performance.
On the other hand, the dependence bias indicates that the dependence is a tradable resource, i.e., when the dependence resource is utilized another resource can be saved. For example, the negative dependence can be traded for transmission power in terms of the performance measures.

\end{abstract}


\IEEEpeerreviewmaketitle

\section{Introduction}

In mathematics, the stochastic dependence is a property of the dependent elements, specified by the probability measure, and independence is a special case with a product measure of probability.
The dependence scenario, which is probably uncertain or is intractable to get an explicit mathematical expression, raises additional analytical issues that differ from the independence scenario.
In real world, the dependence corresponds to the interrelationship of the system states through time and space, and different forms of dependence result in different system performances \cite{sun2017statistical}. 
In other words, the stochastic dependence is not only a mathematical property but also a physical resource.
Considering the diverse characteristics and distinguishing effects of the stochastic dependence, 
it is intriguing to study how to control the dependence in a system in order to obtain an improved performance.
Particularly, a theory of dependence control is built in \cite{sun2018hidden}, and it treats the system as a functional of controllable and uncontrollable random parameter processes and it proves that a manipulation of the dependence in a controllable random parameter process has a consequence on the overall system performance.

In this paper, we further study the dependence control theory, with respect to both the marginal distributions and the dependence structures \cite{ruschendorf2013mathematical}, in the context of the wireless channel capacity, and we obtain two sets of results, namely tail domination and dependence bias, which are respectively about the light-tail property of the wireless channel capacity and the dependence influence on the dependence control mechanism.

\subsection{Tail Domination}

The wireless signals are electromagnetic radiations and the signal propagation environment is a passive medium with dissipation that is the loss of field energy due to absorption, and dispersion that is the variation of the refractive index in the medium \cite{klauder2005signal}\cite{parsons1992mobile}\cite{jakes1994microwave}.
The dissipation causes the energy loss of the signals on the path from the transmitter to the receiver \cite{parsons1992mobile}.
This effect is termed as the large-scale fading \cite{rappaport2001wireless}.
The dispersion causes the reflection, diffraction, and scattering of the transmitted signals \cite{parsons1992mobile}, which result in the multipath interference and the Doppler shift of the received signals.
This effect is termed as the small-scale fading \cite{rappaport2001wireless}.
As a characterization of the propagation channel, the channel gain is defined by the ratio of of the receiver-to-transmitter power, of which the reciprocal is defined as the channel loss.
As a result of the energy conservation law, the channel gain is less than one or the channel loss is greater than one.

Consider the multiple-input-multiple-output channel model that is expressed as \cite{tse2005fundamentals}
\begin{equation}
\mathbf{y} (t) = \bm{H}(t) \mathbf{x}(t) + \mathbf{w}(t),\ t\in\mathbb{N},
\end{equation}
where $\mathbf{x}(t) \in \mathbb{C}^{N_T}$, $\mathbf{y}(t) \in \mathbb{C}^{N_R}$, $\mathbf{w}(t)\sim\mathcal{C}\mathcal{N}\qty(0,N_0 \bm{I}_{N_R})$, and $\bm{H}(t) \in \mathbb{C}^{N_R \times N_T}$ is the channel gain matrix.
For simplification, we omit the time index.
The instantaneous capacity $c \in \mathbb{R}$ is defined by the mutual information, which is a function $f: \mathbb{R}\times \mathbb{C}^{N_R \times N_R} \rightarrow \mathbb{R}$ of the product of the transmission power $p$ and the channel matrix $\bm{H}\bm{H}^\ast$, i.e., 
\begin{equation}
f: p\bm{H}\bm{H}^\ast \mapsto c,
\end{equation}
where we treat the instantaneous power as a random variable.
Specifically, if the tail distribution function satisfies \cite{asmussen2010ruin} $\overline{F}_X(x) = O\qty(e^{-\theta x}),\ \exists \theta>0$, equivalently, $\mathbb{E}\qty[e^{\theta X}] <\infty, \exists \theta>0$, then the distribution is light-tailed; otherwise, it is heavy-tailed. 
The heavy-tailed distribution indicates that extreme values occur with a relatively high probability \cite{falk2010laws}.
Particularly, if the tail is super-heavy, it has no finite moments \cite{halliwell2013classifying}, e.g., the distributions with slowly varying tails.
The class of slowly varying functions includes constants, logarithms, iterated logarithms, powers of logarithms \cite{davis2009probabilistic}.

We obtain that the sufficient condition for the light-tail wireless channel capacity is the existence of the mean value of the power law of the product of the random power and the maximum eigenvalue of the channel matrix, i.e., 
\begin{equation}
\overline{F}_{c}(x) = O\qty(e^{-\theta x}),\ \exists \theta>0 
\impliedby
\mathbb{E} \qty[ \qty(p \lambda_{\max})^\theta ] < \infty,\ \exists \theta>0,
\end{equation}
where the right hand side is equivalent to $\mathbb{E} \qty[ \qty(p \Tr\qty[\bm{H}{\bm{H}}^\ast])^\theta ] < \infty,\ \exists \theta>0$, in terms of the tail behavior, they are equivalently expressed as $\overline{F}_{p\lambda_{\max}}(x) = O\qty(x^{-\theta x}),\ \exists \theta>0$, and $\overline{F}_{p\Tr\qty[\bm{H}{\bm{H}}^\ast]}(x) = O\qty(x^{-\theta x}),\ \exists \theta>0$.
Specifically, $p=1$ corresponds to the deterministic power scenario.
In addition, for the broadband channel scenario, the channel matrix is the diagonal matrix of each sub-channel matrices, i.e., $\bm{\mathcal{H}} = \text{diag}\qty{ \bm{H}_1, \ldots, \bm{H}_N }$.

We provide the following observations, which largely explain the light-tail property of the wireless channel capacity.

\begin{itemize}[wide, labelwidth=!, labelindent=0pt]
\item
It is interesting to note that the typical large-scale fading distribution is heavy-tailed, e.g., the Lognormal distribution, while the typical small-scale fading distribution is light-tailed, e.g., the Rayleigh, Rice, and Nakagami distributions. 
Specifically, if a random variable is lognormal, then its reciprocal is also lognormal. 
The tail property indicates that the large-scale fading effects, like path loss and shadowing, are more likely to cause large values of both channel loss and gain, which may be due to the large shadow dynamics in the propagation environment; while the small-scale fading effects, like the multipath interference and Doopler shift, are less likely to cause large values of channel gain or the random values are more likely to be concentrated around the mean.
Since both light-tailed and heavy-tailed distributions with finite mean are used to model the channel gain, the parametric distributions that can model both heavy-tailed and light-tailed distributions are of interest, e.g., the Weibull distribution \cite{sagias2005gaussian}\cite{pavlovic2013statistics}.
These theoretical insights on the stochastic models match the empirical results \cite{jakes1994microwave}.
The restriction that the passive channel gain is less than one exclude the existence of fading models with super-heavy tails.
In addition, since the random values of the stochastic models, whether the light-tailed distribution or the heavy-tailed distribution, are unbounded, the stochastic models of the wireless channels are strictly not passive systems \cite{klauder2005signal}\cite{mammela2006normalization}, because of the violation of the energy conservation law.

\item
Though the wireless system can be energy unlimited \cite{mammela2010relationship}, the transmission power is unlikely to have an infinite mean, thus, the tail of the power distribution is lighter than the super-heavy distribution.
When there are active relays in the wireless channels, the whole channel gain is the product of each individual channel gain. 
However, the tail of the product distribution can be asymptotically bounded above and below by the tail of a dominating random variable of the product for both independence and dependence scenarios \cite{yang2011tail}\cite{chen2018extensions}\cite{jiang2011product}\cite{yang2013subexponentiality}.
In addition, the gain saturation also exclude the possibility of unlimited gain in active medium \cite{nistad2008causality}. 
Thus, the whole channel gain is more likely to have a tail behavior lighter than the super-heavy tail. 
On the other hand, when the power in the capacity formula is set to be deterministic, e.g., the mean value of power, normalization is usually considered for the channel matrix.
Specifically, if the channel description is based on the average transmitter power ${P}_T$ \cite{paulraj2003introduction}, then, the channel matrix $\bm{H}$ is non-normalized; 
and if the description uses the average receiver power ${P}_R$, 
then the channel matrix $\overline{\bm{H}}$ is normalized \cite{foschini1998limits}\cite{telatar1999capacity}.
Mathematically, it is expressed as \cite{foschini1998limits} ${P}_T^{1/2} \cdot \bm{H} = {P}_R^{1/2} \cdot \overline{\bm{H}}$.
For example, the normalized channel gain of the Rayleigh fading channel is \cite{foschini1998limits}\cite{tse2005fundamentals} $\overline{H}_{ij} \sim \mathcal{C}\mathcal{N}(0,1)$ and $\mathbb{E}\qty[\overline{H}_{ij} \overline{H}_{ij}^\ast] = 1$.
The normalization indicates that the mean values of the matrix identities exist, which excludes the existence of the fading models with super-heavy tails. 
\end{itemize}

In all, for the typical stochastic channel models and the power supply systems in practice, the distribution of the capacity, which is a logarithm transform of the product of the fading effects and random power, is light-tailed, because the logarithm function transforms a less than super-heavy distribution to a light-tailed distribution.

\subsection{Dependence Bias}

We treat the wireless channel capacity as a functional of random parameters \cite{sun2017statistical}\cite{sun2018hidden}, which are either uncontrollable or controllable, the uncontrollable parameters represent the property of the environment that can not be interfered, e.g., fading, and the controllable parameters represent the configurable property of the wireless system, e.g., power.
We specify that the cardinality of the parameter set $\left( X_t^1, X_t^2, \ldots, X_t^n \right)$ is time-invariant and the function $f_t : \mathbb{R}^{n} \to \mathbb{R}$ is time-variant,
i.e.,
\begin{equation}
X_t = f_{t} \left( X_t^1, X_t^2, \ldots, X_t^n \right).
\end{equation}
We study how to transform the dependence in the functional process $\qty{ X_t }$, by manipulating the dependence in parameter processes $\qty{ X_t^i }$, $1\le i\le n$.
There are two ways to implement this dependence transform, i.e., one by transforming the dependence structure from the positive dependence to the negative dependence, and the other by transforming the marginal distributions.

This functional specification is extensible to the general stochastic process on the Polish space, i.e., the stochastic process as a function of a set of random parameters, each of which is itself a stochastic process, in other words, we treat the stochastic process as a functional of a multivariate stochastic process and the functional maps the multivariate stochastic process to a univariate stochastic process.
For example, this functional perspective is useful for studying the dependence impact of an individual arrival process on the aggregation of a set of multiplexed arrival processes.

We highlight the following results, which provide guidelines for dependence control.
\begin{itemize}[wide, labelwidth=!, labelindent=0pt]
\item
The dependence is a resource that can be traded off, i.e., when the dependence is utilized, another form of resource can be saved, e.g., more amounts of negative dependence can exchange for less amounts of transmission power.
The chain relation, $\bm{X} \le_{sm} \widetilde{\bm{X}} \implies \sum_{j=1}^{t} {X}_j \le_{cx} \sum_{j=1}^{t} \widetilde{X}_j \implies \mathbb{E} \sum_{j=1}^{t} {X}_j = \mathbb{E} \sum_{j=1}^{t} \widetilde{X}_j$, means the supermodular order of the dependence structures implies the convex order of the variability of the partial sum with equal mean.
To take into account both the mean and the variability, we use the increasing convex order for further elaboration.
Specifically, the mean and the variability are exchangeable for each other, i.e., if the variability is relatively small, then a relatively greater mean can be tolerated while satisfying the increasing convex order, vice versa.
The mathematical expressions are as follows, if $X \le_{icx} Y$ and $\mathbb{E} X \le \mathbb{E} Z^\prime \le \mathbb{E} Y$, then it is possible that $Z^\prime \le_{icx} Y$, because we have $X \le_{icx} Y \iff X \le_{st} Z \le_{cx} Y$ \cite{shaked2007stochastic}; and if $X \le_{icx} Y$, then $ X \le_{cx} Z^\prime \le_{st} Y$ such that $Z^\prime \le_{icx} Y$, because we have $X \le_{icx} Y \iff X \le_{cx} Z \le_{st} Y$ \cite{shaked2007stochastic}.
Complementary results hold in the sense of the increasing concave order \cite{shaked2007stochastic}.

\item
When the backlog and the delay are used as the performance measures \cite{sun2018hidden}, the arrival process and the service process are consistent in the manipulation of the dependence strength and are different in the manipulation of the marginals.
Specifically, for the manipulation of the dependence,
the objective is the convex ordering $\sum_{j=1}^{t} {X}_j \le_{cx} \sum_{j=1}^{t} \widetilde{X}_j$, where $X_j = a_j - s_j$ represents the instantaneous arrival amount minus the instantaneous service amount,
while for the manipulation of the marginals, 
the objective for the arrival process is the increasing convex ordering $\sum_{j=1}^{t} {X}_j \le_{icx} \sum_{j=1}^{t} \widetilde{X}_j$ and the objective for the service process is the increasing concave ordering $-\sum_{j=1}^{t} {X}_j \ge_{icv} -\sum_{j=1}^{t} \widetilde{X}_j$.
This is coherent with the intuition that a smaller and less variable arrival process or a greater and less variable service process leads to a better system performance in terms of the backlog and delay \cite{sun2018hidden}.

\item
The manipulation of the marginal distributions has a dependence bias, while the manipulation of the dependence structure fixing the marginals has no such dependence bias.
Specifically, the dependence bias means that, if a parameter process bears negative dependence, then the manipulation of each individual marginals with respect to the (increasing) convex order can not lead effectively to the (increasing) convex order of the partial sums, i.e., the (increasing) convex order of the marginals implies the (increasing) convex order of the partial sum holds for positive dependence and not for negative dependence \cite{muller2002comparison}. 
The dependence bias of the marginals provides an opportunity for dependence control. Specifically, the dependence bias means that the increasing convex order of the partial sum is insensitive to the marginal manipulation of the parameter process with negative dependence, e.g., the increasing convex order still holds for a partial sum with smaller mean values of the marginals. In other words, a better system performance, in terms of backlog and delay, can be achieved in the scenario of negative dependence in the processes, even with a smaller mean value of the service process or a greater mean value of the arrival process. 
\end{itemize}

\subsection{Related Work}
There are some related work in the literature and the comparisons with this paper are as follows.
{\bf (i)}
The light-tailed property of wireless channel capacity for the single-input-single-output channel is investigated in \cite{sun2017statistical}. In this paper, we reason why the light-tailed behavior is an intrinsic property of the wireless channel capacity, and we extend the results to multiple-input-multiple-output channel, with an extensive study on the equivalent conditions and sufficient conditions regarding the statistical identities of the wireless channel.
The statistical property of wireless channel capacity is studied in \cite{rafiq2011statistical}, where the focus is on the distribution functions and first and second order statistics of the capacity, the detailed fading distributions are used, e.g., Rayleigh, Rice, and Nakagami. In this paper, instead of the exact fading distributions, we use the light-tailed and heavy-tailed distribution classes and show that the distribution of the capacity based on these typical fading distributions is intrinsically light-tailed. Thus, the results in this paper is more general than \cite{rafiq2011statistical} and indicate more possibilities of wireless channel modeling, i.e., more distributions as alternative to the typical fading models. 
{\bf (ii)}
The tail asymptotic is investigated in \cite{sarantsev2011tail} for the product and sum of random variables in terms of the asymptotic equality $f(x)\sim g(x)$. In this paper, we extend the analysis to specific heavy-tailed and light-tailed distribution classes, e.g., the long-tail distribution, the regular varying distribution, and the light-tailed distribution, specifically, we find that the slowly varying distribution can dominate the tail behavior for the sum and product distribution, 
moreover, we extend the analysis beyond the asymptotic equality to more asymptotic notations, e.g., $f(x) = O\qty(g(x))$, $f(x)=\Theta\qty(g(x))$, $f(x) = \omega\qty(g(x))$, and $f(x) = o\qty(g(x))$.
Since the capacity is a logarithm transform of the product of the power and the fading random variable, the less strict asymptotic bound provides more flexibility than the asymptotic equality, i.e., it has less restriction and can capture more distribution scenarios, most importantly, it is sufficiently enough to investigate the light-tail behavior that is defined by the asymptotic bound $f(x) = O\qty(g(x))$.
Another related work is \cite{tang2008light}, which provides conditions for the product of a light-tailed random variable and a heavy-tailed random variable to be heavy-tailed.
In contrast to the result with asymptotic precision up to some distribution classes in \cite{tang2008light}, we show results of the exact tail domination with respect to a certain distribution function in this paper.
{\bf (iii)}
The dependence control theory is studied in \cite{sun2018hidden}, where the stochastic dependence is not only treated as a mathematical property but also as a physical resource \cite{sun2017statistical}, the benefits of utilizing the stochastic dependence resource are elaborated, and the dependence transformability is proved with respect to the supermodular order. In this paper, we extend the results to the (increasing) directionally convex order and the usual stochastic order, particularly, the (increasing) directionally convex order takes into account the impacts of both the marginals and dependence structures in comparison, while the supermodular order requires identical marginals of the compared stochastic processes.
Such extension provides additional insights into the dependence control theory as regards the dependence manipulation and marginal manipulation.

The rest of this paper is structured as follows.
The tail property of the MIMO (multiple-input-multiple-output) channel is studied in Sec. \ref{tail-property}.
The dependence transform of stochastic processes is studied in Sec. \ref{transformability}.
Finally, this paper is concluded and future work are discussed in Sec. \ref{conclusion}.

\section{Tail Property}
\label{tail-property}

Let $(\Omega, \mathscr{F}, \Pr)$ be a probability space and $\bm{X}: \Omega \to \mathbb{C}^{m\times n}$ be measurable with respect to $\mathscr{F}$ and the Borel $\sigma$-algebra on $\mathbb{C}^{m \times n}$.
Denote $\bm{ \mathfrak{X} } = \qty{ \bm{X} \in \mathbb{C}^{n\times n}: \bm{X} = \bm{X}^\ast }$, where $\ast$ represents the conjugate transpose.
Denote the cone \cite{ahlswede2002strong} $\bm{ \mathfrak{X}} _{\ge 0}=\qty{ \bm{X} \in \bm{\mathfrak{X}}: \bm{X} \ge 0 }$, which introduces a partial order in $\bm{\mathfrak{X}}$, i.e., $\bm{X} \ge 0$ is equivalent to that all the eigenvalues of $\bm{X}$ are nonnegative. Similarly, $\bm{ \mathfrak{X}} _{> 0}=\qty{ \bm{X} \in \bm{\mathfrak{X}}: \bm{X} > 0 }$.

\subsection{Deterministic Power Fluctuation}

Consider the flat fading MIMO channel $\bm{H} \in \mathbb{C}^{ N_R\times N_T }$, $\bm{H}{\bm{H}}^\ast \in \bm{\mathfrak{X}}_{\ge 0}$. 
The capacity, in bits per second, under total average transmit power constraint, is expressed as \cite{paulraj2003introduction}
\begin{equation}
c = W \max_{\Tr[\bm{R}_{\bm{s}\bm{s}}] = N_T } \log_2 \det \qty( \bm{I}_{N_R} + \frac{\rho}{N_T} \bm{{H}} \bm{R}_{\bm{{s}{s}}} \bm{{H}}^\ast ),
\end{equation}
where $W$ is the bandwidth, $\rho = \frac{P}{N_0 W}$, $P$ is the total average transmit power, $N_0$ is the noise power spectral density, $\bm{R}_{\bm{{s}{s}}} = \mathbb{E}[\bm{{s}}\bm{{s}}^\ast]$ is the covariance matrix for the transmitted signal $\bm{{s}} \in \mathbb{C}^{N_T \times 1}$.

The frequency-selective fading channel formulation requires a block diagonal extension of the flat fading channel model.
The capacity, in bits per second, under total average transmit power constraint, is expressed as \cite{paulraj2003introduction}
\begin{equation}
c = \frac{W}{N} \max_{\Tr[\bm{R}_{\bm{\mathcal{S}\mathcal{S}}}] = N_T N} \log_2 \det \qty( \bm{I}_{N_R N} + \frac{\rho}{N_T} \bm{\mathcal{H}} \bm{R}_{\bm{\mathcal{S}\mathcal{S}}} \bm{\mathcal{H}}^\ast ),
\end{equation}
where $W$ is the bandwidth, $\rho = \frac{P}{N_0 W}$, $P$ is the total average transmit power, $N_0$ is the noise power spectral density, $N$ is the number of sub-channels, $\bm{\mathcal{H}} \in \mathbb{C}^{N_T N \times N_R N}$ is the block diagonal matrix with $\bm{H}_i$ as the block diagonal elements, and $\bm{R}_{\bm{\mathcal{S}\mathcal{S}}} = \mathbb{E}[\bm{\mathcal{S}}\bm{\mathcal{S}}^\ast]$ is the covariance matrix for the transmitted signal $\bm{\mathcal{S}} = [\bm{s}_1^T, \ldots, \bm{s}_N^T]^T \in \mathbb{C}^{N_T N \times 1}$.

\begin{remark}
The identity matrix $\bm{I}$ in the capacity formula implies that the capacity is non-negative, i.e., $c: \Omega \rightarrow \mathbb{R}_{\ge 0}$.
\end{remark}

\begin{remark}
The typical stochastic models of the channel gain are the Rayleigh, Rice, and Nakagami distributions \cite{goldsmith2005wireless}. 
The shadowing model is the Lognormal distribution \cite{rappaport2001wireless}\cite{goldsmith2005wireless}, which is able to superimpose the path loss.
\end{remark}

\begin{remark}
If $ \log Y \sim N\qty(\mu, \sigma^2)$, then $ a + b \log Y \sim N\qty(a+ b \mu, b^{2} \sigma^2)$, where $a, b \in \mathbb{R}$, thus, if $Y$ is lognormal, then $a Y^{b}$ is also lognormal in general.
This result explains the product form of the combined effect of the multiple path interference, shadowing, and path loss \cite{patzold2011mobile}.
\end{remark}

\begin{remark}
Since the normal distribution with zero mean is symmetric, we have the equal lognormal distributions $Y^{-1} \overset{d}{=} Y$, because of $- \log Y \overset{d}{=} \log Y \sim N\qty(0, \sigma^2)$.
This result implies that the quotient $X/Y$ of an arbitrary random variable $X$ with a lognormal random variable $\log Y \sim N\qty(0, \sigma^2)$, where $X$ and $Y$ are independent, equals in distribution the product $XY$ of the two random variable, i.e., $X/Y \overset{d}{=} XY$.
This relation does not hold for the general normal distribution.
\end{remark}

\begin{remark}
For a almost surely positive random variable, $X$, the right tail behavior of $1/X$, $\mathbb{P}(1/X > x) = O\qty(e^{-\theta x}),\ \exists \theta>0$, and $\mathbb{P}(1/X > x) = O\qty(x^{-\theta }),\ \exists \theta>0$, corresponds to left tail behavior of $X$ \cite{gulisashvili2016tail}\cite{asmussen2016exponential}, $\mathbb{P}\qty(X < 1/x) = O\qty(e^{-\theta x}),\ \exists \theta>0$, and $\mathbb{P}\qty(X < 1/x) = O\qty(x^{-\theta}),\ \exists \theta>0$, i.e., $\limsup\limits_{y\rightarrow 0} \frac{\mathbb{P}\qty(X < y)}{ e^{-\theta/ {y}} } < \infty$, $\exists \theta>0$, and $\limsup\limits_{y\rightarrow 0} \frac{\mathbb{P}\qty(X < y)}{ y^{\theta} } < \infty$, $\exists \theta>0$.
Letting $Y=1/X$, we obtain the complementary results.
Considering the reciprocal relation between the channel loss $\phi = P_T/P_R$ and channel gain $\psi = P_R/P_T$, both the right tail and the left tail matter for the stochastic channel models.
\end{remark}

We present some equivalence results of the function of random variables.
The proof is shown in Appendix \ref{proof-lemma-tail-equivalence}.

\begin{lemma}\label{lemma-tail-equivalence}
Consider a flat MIMO channel $\bm{H} \in \mathbb{C}^{N_{R}\times N_T}$. 
The capacity is upper bounded by $c =a\log_{2}(1+b\lambda_{\max})$, where $a, b\in \mathbb{R}_{>0}$ and $\lambda_{\max}$ is the maximum eigenvalue of $\bm{H}\bm{H}^\ast$.
\begin{enumerate}[wide, labelwidth=!, labelindent=0pt]
\item
For the tail property, we have the equivalent results
\begin{multline}
\overline{F}_{c}(x)=O(e^{-\theta x}),\ \exists \theta>0 \\
\iff \overline{F}_{\lambda_{\max}}(x)=O(x^{-\theta}),\ \exists \theta>0 \\
\iff \overline{F}_{\Tr\qty[\bm{H}\bm{H}^\ast]}(x)=O(x^{-\theta}),\ \exists \theta>0.
\end{multline}

\item
For the power law function of the maximum eigenvalue, we have the equivalent expressions 
\begin{multline}
\mathbb{E} \qty[ \qty( 1 + {\Delta}  \lambda_{\max}  )^\theta ] < \infty,\ 0< \Delta <\infty, \ \exists \theta>0 \\
\iff
\mathbb{E} \qty[ \qty( 1 + \lambda_{\max}  )^\theta ] < \infty, \ \exists \theta>0 \\
\iff
\mathbb{E} \qty[ \qty( \lambda_{\max}  )^\theta ] < \infty, \ \exists \theta>0.
\end{multline}
Specifically, $\theta=1$ corresponds to $\mathbb{E}\qty[\lambda_{\max}]<\infty$. In addition, we have $\mathbb{E}\qty[ \qty(\lambda_{\max})^{\theta_1}] = \infty \implies \mathbb{E}\qty[ \qty( \lambda_{\max}  )^{\theta_2} ] = \infty$, $\forall 0< \theta_1 < \theta_2$. 

\item
In addition, we have another pair of equivalent expressions for the exponential function of the eigenvalue, i.e.,
\begin{multline}
\mathbb{E} \qty[ e^{\theta  \Tr\qty[ \bm{H}\bm{H}^\ast ]} ] < \infty, \ \exists \theta >0 \\
\iff
\mathbb{E} \qty[ e^{\theta  \lambda_{\max}} ] < \infty, \ \exists \theta >0. 
\end{multline}
\end{enumerate}
\end{lemma}

\begin{remark}
The parameters $\theta$ in the two equations, $\overline{F}_{c}(x)=O(e^{-\theta x})$ and $\overline{F}_{\lambda_{\max}}(x)=O(x^{-\theta})$, are not necessarily equal.
\end{remark}

\begin{remark}
The above equivalent results indicate that, for $X\in \mathbb{R}_{\ge 0}$, $\mathbb{E}\qty[ X^\theta ]<\infty,\ \exists \theta>0 \iff \overline{F}_X(x)=O\qty(x^{-\theta}),\ \exists \theta>0$.
However, it is interesting to notice that, for a common $\theta>0$ and $X\in \mathbb{R}_{\ge 0}$, we only have $\mathbb{E}\qty[ X^\theta ]<\infty \implies \overline{F}_X(x)=O\qty(x^{-\theta})$, and the reverse does not hold in general.
Because it is shown in \cite{sarantsev2011tail} that $\mathbb{E}\qty[ X^\theta ]<\infty$, where $\theta>0$ and $X$ is a nonnegative random variable, if and only if $\overline{F}_X(x)=o\qty(x^{-\theta})$ and $\int_{0}^{\infty} \overline{F}_X(x) x^{\theta-1} d x <\infty $.
\end{remark}

\begin{remark}
The mean identity $\mathbb{E}[X^\theta]$, for $X: \Omega \rightarrow \mathbb{R}_{\ge 0}$ and $\theta>0$, is a special case of Mellin-Stieltjes transform \cite{zolotarev1957mellin} of the distribution function $F_X(x)$.
\end{remark}

\begin{remark}
Suppose $X$ is a regularly varying non-negative random variable with index $\alpha > 0$. Then \cite{embrechts1997modelling}, $\mathbb{E}\qty[ X^\beta ]<\infty$, for $\beta<\alpha$; and $\mathbb{E}\qty[ X^\beta ] = \infty$, for $\beta>\alpha$. 
\end{remark}

\subsubsection{Arbitrary Channel Side Information}

We present the sufficient condition for the light-tailed capacity of the flat channel.
We present the proof in Appendix \ref{proof-theorem-tail-sufficient}.

\begin{theorem}\label{theorem-tail-sufficient}
Consider a flat MIMO channel $\bm{H} \in \mathbb{C}^{N_{R}\times N_T}$. 
If the mean identity exists, i.e., 
\begin{equation}
 \mathbb{E} \qty[ \qty( 1 +  \lambda_{\max}  )^\theta ] < \infty, \ \exists \theta>0,
\end{equation}
where $\lambda_{\max}$ is the maximum eigenvalue of $\bm{H} \bm{H}^\ast$,
the distribution of the capacity of the MIMO channel (with or without full channel side information) is light-tailed.
\end{theorem}

We present the sufficient condition for the light-tailed capacity of the frequency-selective channel. 
The proof is shown in Appendix \ref{proof-theorem-frequency-selective-sufficient}.
 
\begin{theorem}\label{theorem-frequency-selective-sufficient}
Consider a frequency-selective MIMO channel with sub-channel $\bm{H}_i \in \mathbb{C}^{N_{R}\times N_T}$, $i\in\qty{1,\ldots, N}$, and block diagonal matrix $\bm{\mathcal{H}} = \text{diag} \qty(\bm{H}_1,\ldots, \bm{H}_N)$. 
For the scenarios where channel side information known or unknown at the transmitter,
if the mean identity exists for each sub-channel, i.e., 
\begin{equation}
 \mathbb{E} \qty[ \qty( 1 + \lambda_{\max}^{i}  )^\theta ] < \infty, \ \exists \theta>0,\ \forall i\in\qty{1,\ldots, N},
\end{equation}
where $\lambda^{i}_{\max}$ is the maximum eigenvalue of $\bm{H}_i \bm{H}_i^\ast$, or the equivalent condition is satisfied, i.e., 
\begin{equation}
\mathbb{E} \qty[ \qty( 1 + \lambda_{\max}  )^\theta ] < \infty, \ \exists \theta>0,
\end{equation}
where $\lambda_{\max}$ is the maximum eigenvalue of $\bm{\mathcal{H}}\bm{\mathcal{H}}^\ast$, the distribution of the capacity is light-tailed.
\end{theorem}

\begin{remark}
The equivalent expression of the sufficient condition means that it is equivalent to consider the block diagonal matrix of the frequency-selective channel as a whole or to consider the matrix of each sub-channel individually.
\end{remark}

\begin{remark}
If the trace identity exists, i.e., $\mathbb{E} \qty [ \Tr \qty[ e^{\theta \bm{\mathcal{H}}  \bm{\mathcal{H}}^\ast  } ] ] < \infty,\ \exists \theta>0$, then the maximum eigenvalue distribution and the capacity distribution are light-tailed.
Because the maximum eigenvalue distribution is exponentially bounded \cite{tropp2015introduction}
\begin{equation}
\Pr \qty( \lambda_{\max} (\bm{X}) \ge x ) \le  e^{-\theta x} \cdot \mathbb{E} \qty [ \Tr \qty[ e^{\theta \bm{X} } ] ], \ \forall \theta>0,
\end{equation}
where $\bm{X} = \bm{\mathcal{H}}\bm{\mathcal{H}}^\ast \in \bm{\mathfrak{X}}$.
Since the matrix $\bm{\mathcal{H}}  \bm{\mathcal{H}}^\ast $ is block diagonal \cite{higham2008functions}, $e^{\theta \bm{\mathcal{H}}  \bm{\mathcal{H}}^\ast} = \text{diag} \qty(e^{\theta \bm{H}_1 \bm{H}_1^\ast  }, \ldots, e^{\theta \bm{H}_N \bm{H}_N^\ast  }) $, and
\begin{IEEEeqnarray}{rCl}
\mathbb{E} \qty [ \Tr \qty[ e^{\theta \bm{\mathcal{H}}  \bm{\mathcal{H}}^\ast  } ] ] = \sum_{i=1}^{N} \mathbb{E} \qty [ \Tr \qty[ e^{\theta \bm{{H}}_i  \bm{{H}}_i^\ast  } ] ].
\end{IEEEeqnarray}
Thus, $\mathbb{E} \qty [ \Tr \qty[ e^{\theta \bm{\mathcal{H}}  \bm{\mathcal{H}}^\ast  } ] ] < \infty,\ \exists \theta>0$ entails $\mathbb{E} \qty [ \Tr \qty[ e^{\theta \bm{H}_i \bm{H}_i^\ast  } ] ] < \infty$, $\exists \theta>0$, $\forall i\in\qty{1,\ldots, N}$, which is non-negative.
\end{remark}

\subsubsection{Without Channel Side Information at Transmitter}

We present the sufficient and necessary condition for the flat channel capacity distribution to be light-tailed, when the channel side information is not known at the transmitter.
We present the proof in Appendix \ref{proof-theorem-s-n-csi-transmitter}.

\begin{theorem}\label{theorem-s-n-csi-transmitter}
Consider that the channel side information is only known at the receiver.
The capacity distribution of the flat channel is light tailed, if and only if the tail of the determinant term is expressed as
\begin{equation}
\mathbb{E} \qty[ \qty( \det \qty( \bm{I}_{N_R} +  \bm{\Lambda} ) )^\theta ] < \infty,\ \exists \theta >0,
\end{equation}
where $\bm{H}\bm{H}^\ast = \bm{Q}\bm{\Lambda}\bm{Q}^\ast$, $\bm{Q}\bm{Q}^\ast=\bm{Q}^\ast \bm{Q} = \bm{I}_{N_R}$, and $\bm{\Lambda}=\text{diag}\qty{\lambda_1, \ldots, \lambda_{N_R}}$, $\lambda_i\ge 0$.
Equivalently, the condition is expressed as
\begin{equation}
\mathbb{E} \qty[  \prod_{i=1}^{r(\bm{H})} \qty( 1+ \lambda_i )^\theta  ] < \infty,\ \exists \theta >0,
\end{equation}
where $0<\lambda_i \in \bm{\Lambda}$ and $r(\bm{H})$ is the rank of $\bm{H}$.
\end{theorem}

\begin{remark}
The conditions are equivalently expressed as 
$\mathbb{E} \qty[  \prod_{i=1}^{r(\bm{H})} \qty( \lambda_i )^\theta  ] < \infty,\ \exists \theta >0$.
The proof is similar to that of Lemma \ref{lemma-tail-equivalence}.
\end{remark}

\begin{remark}
Considering the inequality $  \prod_{i=1}^{r(\bm{H})} \qty( 1+  \lambda_i )^\theta  \le  \qty( 1+ \lambda_{\max} )^{\theta {r(\bm{H})}}  $, where $\lambda_{\max} = \max\limits_{1\le i \le r(\bm{H})}{\lambda_i}$, the sufficient and necessary condition relaxes to a sufficient condition, i.e.,
\begin{multline}
\mathbb{E} \qty[  \prod_{i=1}^{r(\bm{H})} \qty( 1+ \lambda_i )^\theta  ] < \infty,\ \exists \theta >0 \\
\impliedby
\mathbb{E} \qty[ \qty( 1+ \lambda_{\max} )^{\theta } ] < \infty, \ \exists \theta > 0，
\end{multline}
which is equivalent to $\mathbb{E} \qty[ \qty( \lambda_{\max} )^{\theta } ] < \infty, \ \exists \theta > 0$.
\end{remark}

\begin{remark}
Considering the Fredholm determinant \cite{gohberg2012traces}, for $|z|$ small enough, $\log\det \qty( \bm{I} + z \bm{\Lambda} ) = \Tr \log\qty( \bm{I} + z \bm{\Lambda} ) = \sum_{k=1}^{\infty}  \frac{ (-1)^{k+1} }{k} z^k \Tr[\bm{\Lambda}^k] $, the condition is alternatively expressed as
\begin{equation}
\mathbb{E} \qty[ e^{ \theta \sum_{k=1}^{\infty}  \frac{ (-1)^{k+1} }{k} \qty(\frac{\rho}{N_T})^k {\Tr \qty[\bm{\Lambda}^k ]} } ] < \infty, \ \exists \theta>0,
\end{equation}
where $\bm{\Lambda} = \bm{H}\bm{H}^\ast$ or $\bm{H}\bm{H}^\ast = \bm{Q}\bm{\Lambda}\bm{Q}^\ast$.
Specifically, for $\bm{H}\bm{H}^\ast = \bm{Q}\bm{\Lambda}\bm{Q}^\ast$, we have ${\Tr \qty[\bm{\Lambda}^k ]} = \sum_{i=1}^{r(\bm{\Lambda})} (\lambda_{i}(\bm{\Lambda}))^{k} $.

For arbitrary $z$, according to the Plemelj-Smithies formulas \cite{gohberg2012traces}, we have
$\det( \bm{I} + z \bm{\Lambda} ) = 1 + \sum_{k=1}^{r(\bm{\Lambda})} \frac{d_k(\bm{\Lambda})}{k!} z^k $, thus the condition is expressed as
\begin{equation}
\mathbb{E} \qty [ \qty(1 + \sum_{k=1}^{r(\bm{\Lambda})} \frac{d_k(\bm{\Lambda})}{k!} \qty(\frac{\rho}{N_T})^k)^\theta ] < \infty,\ \theta>0,
\end{equation}
where
\begin{multline}
d_k(\bm{\Lambda}) = \\
\mqty | 
\Tr \bm{\Lambda} & k-1 & 0 & 0 & \ldots & 0 & 0 \\ 
\Tr \bm{\Lambda}^2 & \Tr \bm{\Lambda} & k-2 & 0 & \ldots & 0 & 0 \\ 
\vdots & \vdots & \vdots & \vdots & \ddots & \vdots & \vdots \\ 
\Tr \bm{\Lambda}^{k-1} & \Tr \bm{\Lambda}^{k-2} & \Tr \bm{\Lambda}^{k-3} & \vdots & \ldots & \Tr \bm{\Lambda} & 1 \\ 
\Tr \bm{\Lambda}^{k} & \Tr \bm{\Lambda}^{k-1} & \Tr \bm{\Lambda}^{k-2} & \vdots & \ldots & \Tr \bm{\Lambda}^2 & \Tr \bm{\Lambda} \\ 
|. \nonumber
\end{multline}
\end{remark}

We present a sufficient condition for the light-tailed property of the frequency-selective channel capacity.

\begin{theorem}
Consider that the channel side information is only known at the receiver.
The capacity distribution of the frequency-selective channel is light tailed, if
\begin{equation}
\mathbb{E} \qty [ \prod_{i=1}^{r(\bm{H_j})} \qty( 1+ \lambda_i^j )^\theta ] < \infty,\ \exists \theta>0, \ \forall j\in\{1, \ldots, N\},
\end{equation}
where $\bm{H}_j$ is the channel model of each sub-channel and $\lambda_i^j$ is the corresponding eigenvalue of $\bm{H}_j\bm{H}_j^\ast$.
\end{theorem}

\begin{proof}
The proof of the frequency-selective channel scenario follows that the light-tailed distribution of the capacity of each sub-channel implies the light-tailed distribution of the overall channel capacity.
\end{proof}

\begin{remark}
The sufficient condition relaxes to $\mathbb{E} \qty[ \qty(1 + \lambda^{j}_{\max})^\theta ] < \infty, \exists \theta>0$, $\forall 1\le j \le N$, where $\lambda^{j}_{\max} = \max\limits_{1\le i \le r(\bm{H}_j)}{ \lambda^j_i }$.
Equivalently, it is expressed as $\mathbb{E} \qty[ \qty(1 + \lambda_{\max})^\theta ] < \infty, \exists \theta>0$, where $\lambda_{\max} = \max\limits_{1\le j\le N}{ \lambda^j_{\max} }$.
\end{remark}

\subsection{Random Power Fluctuation}

We consider the channel scenario, where the channel knowledge is known at the receiver and is unknown at the transmitter, and the transmission power randomly fluctuates over the coherence periods and remains constant in each coherence period.

For the flat fading MIMO channel $\bm{H} \in \mathbb{C}^{ N_R\times N_T }$,
when the transmit power is allocated evenly across the transmit antennas during each coherence period, 
the capacity, in bits per second, is expressed as 
\begin{equation}
c_{p,\bm{H}} = W \log_2 \det \qty( \bm{I}_{N_R} + \frac{1}{N_T N_0 W} p \bm{{H}}  \bm{{H}}^\ast ),
\end{equation}
where $W$ is the bandwidth, $N_0$ is the noise power spectral density, and $p$ is the transmit power that is constant during each coherence period and randomly fluctuates over periods. 
Equivalently, the capacity is expressed as
\begin{IEEEeqnarray}{rCl}
c_{p,\bm{H}} &=& W \sum\nolimits_{i=1}^{r(\bm{H} )} \log_2 \qty( 1 + \frac{1}{N_T N_0 W} p \lambda_i ) \\
&\le& W r(\bm{H} ) \log_2 \qty( 1 + \frac{1}{N_T N_0 W} p \lambda_{\max}  ),
\end{IEEEeqnarray}
where $\lambda_i$ is the eigenvalue of the matrix $\bm{H}\bm{H}^\ast$ and $r(\bm{H})$ is the rank of $\bm{H}$.

\subsubsection{Sufficient Conditions for Light Tails}

We present some preliminary results considering the random power.

\begin{lemma}
Consider a flat MIMO channel $\bm{H} \in \mathbb{C}^{N_{R}\times N_T}$. 
The capacity is upper bounded by $c=a\log_{2}(1+b p \lambda_{\max})$, where $a, b\in \mathbb{R}_{>0}$, $p$ is the random power, and $\lambda_{\max}$ is the maximum eigenvalue of $\bm{H}\bm{H}^\ast$.

\begin{enumerate}[wide, labelwidth=!, labelindent=0pt]
\item
For the tail property, we have the equivalent results
\begin{multline}
\overline{F}_{c}(x)=O(e^{-\theta x}),\ \exists \theta>0 \\
\iff \overline{F}_{p \lambda_{\max}}(x)=O(x^{-\theta}),\ \exists \theta>0 \\
\iff \overline{F}_{p \Tr\qty[ \bm{H}\bm{H}^\ast]}(x)=O(x^{-\theta}),\ \exists \theta>0.
\end{multline}

\item
Alternatively, the tail property is expressed as $\mathbb{E}\qty[ e^{\theta c} ] < \infty$, $\exists \theta>0$, and we have the equivalent expressions
\begin{multline}
\mathbb{E} \qty[ ( 1 + b p \lambda_{\max} )^\theta ] < \infty,\ \exists \theta>0 \\
\iff
\mathbb{E} \qty[ ( 1 + p \lambda_{\max} )^\theta ] < \infty,\ \exists \theta>0 \\
\iff
\mathbb{E} \qty[ ( p \lambda_{\max} )^\theta ] < \infty,\ \exists \theta>0.
\end{multline}
In addition, if $p$ and $\lambda_{\max}$ are independent, then $\mathbb{E} \qty[ ( p \lambda_{\max} )^\theta ] = \mathbb{E} \qty[ p^\theta] \mathbb{E}\qty[ ( \lambda_{\max} )^\theta ] $, and the condition relaxes to $\mathbb{E}\qty[p]<\infty$ and $\mathbb{E}\qty[ \lambda_{\max} ] < \infty$.
\end{enumerate}
\end{lemma}

\begin{proof}
The proof is analog to the proof of the deterministic power scenario.
\end{proof}

We present the sufficient and necessary condition for the light-tailed property of the capacity.

\begin{theorem}
The capacity distribution of the flat channel is light tailed, if and only if the tail of the determinant term is no heavier than the fat tail, i.e., 
\begin{equation}
\mathbb{E} \qty[ \qty( \det \qty( \bm{I}_{N_R} + p \bm{\Lambda} ) )^\theta ] < \infty,\ \exists \theta >0,
\end{equation}
where $\bm{H}\bm{H}^\ast = \bm{Q}\bm{\Lambda}\bm{Q}^\ast$, $\bm{Q}\bm{Q}^\ast=\bm{Q}^\ast \bm{Q} = \bm{I}_{N_R}$, and $\bm{\Lambda}=\text{diag}\qty{\lambda_1, \ldots, \lambda_{N_R}}$, $\lambda_i\ge 0$.
Equivalently, the condition is expressed as
\begin{equation}
\mathbb{E} \qty[  \prod_{i=1}^{r(\bm{H})} \qty( 1+ p \lambda_i )^\theta  ] < \infty,\ \exists \theta >0,
\end{equation}
where $0<\lambda_i \in \bm{\Lambda}$ and $r(\bm{H})$ is the rank of $\bm{H}$.
\end{theorem}

\begin{proof}
The proof is analog to the proof of the deterministic power scenario.
\end{proof}

\begin{remark}
Considering the inequality $  \prod_{i=1}^{r(\bm{H})} \qty( 1+ {p} \lambda_i )^\theta  \le  \qty( 1+ {p} \lambda_{\max} )^{\theta {r(\bm{H})}}  $, where $\lambda_{\max} = \max\limits_{1\le i \le r(\bm{H})} \lambda_i$, the sufficient and necessary condition relaxes to a sufficient condition, i.e.,
\begin{multline}
\mathbb{E} \qty[  \prod_{i=1}^{r(\bm{H})} \qty( 1+ {p} \lambda_i )^\theta  ] < \infty,\ \exists \theta >0 \\
\impliedby
\mathbb{E} \qty[ \qty( 1+ {p} \lambda_{\max} )^{\theta } ] < \infty, \ \exists \theta > 0.
\end{multline}
\end{remark}

We present a set of sufficient conditions for the light-tailed capacity and their relationships.
The proof is in Appendix \ref{proof-theorem-sufficient-map}.

\begin{theorem}\label{theorem-sufficient-map}
Consider a flat MIMO channel $\bm{H} : \Omega \rightarrow \mathbb{C}^{N_{R}\times N_T}$ with random power fluctuation $p: \Omega \rightarrow \mathbb{R}$. 
We have the sufficient condition chain for the light-tailed property of the capacity.
\begin{center}
\begin{tikzpicture}[>=implies]
\matrix (m) [matrix of math nodes, nodes in empty cells, row sep=1em, column sep=1em, 
text height=1.5ex, text depth=0.25ex] 
        { &  & \circled{$0$}  & &     \\
         &  & \circled{$1$}  &  &  \\
       &   \circled{$4$} &   & \circled{$2$} &  \\
      \circled{$5$}  &   &   &  & \circled{$3$} \\
};

\draw[double,<->] (m-2-3) -- (m-1-3);
\draw[double,->] (m-3-2) -- (m-2-3);
\draw[double,->] (m-3-4) -- (m-2-3);
\draw[double,->] (m-4-1) -- (m-3-2);
\draw[double,->] (m-4-5) -- (m-3-4);
\end{tikzpicture}
\end{center}
\begin{IEEEeqnarray}{rCl}
\circled{$0$} &:=& 
\mathbb{E} \qty[ \qty( 1 + p \lambda_{\max}  )^\theta ] < \infty, \ \exists \theta>0  \\
\circled{$1$} &:=&
\mathbb{E} \qty[ \qty( 1 + p \Tr\qty[ \bm{H}\bm{H}^\ast ]  )^\theta ] < \infty, \ \exists \theta>0 \\
\circled{$2$} &:=&
\mathbb{E} \qty[ \qty(  \Tr\qty[ \bm{I} + p \bm{H}\bm{H}^\ast ]  )^\theta ] < \infty, \ \exists \theta>0 \\
\circled{$3$} &:=&
\mathbb{E} \qty[ \qty(  \Tr\qty[ e^{ p \bm{H}\bm{H}^\ast } ]  )^\theta ] < \infty, \ \exists \theta>0 
\\
\circled{$4$} &:=&
\mathbb{E} \qty[ e^{\theta p \lambda_{\max}} ] < \infty,\ \exists \theta>0 \\
\circled{$5$} &:=&
\mathbb{E} \qty [ \Tr \qty[ e^{\theta p \bm{H} \bm{H}^\ast  } ] ] < \infty,\ \exists \theta>0 \\
\circled{$6$} &:=& 
\mathbb{E} \qty[ \qty(  p \lambda_{\max}  )^\theta ] < \infty, \ \exists \theta>0  \\
\circled{$7$} &:=&
\mathbb{E} \qty[ e^{\theta p \Tr\qty[ \bm{H}\bm{H}^\ast ]} ] < \infty, \ \exists \theta >0 
\end{IEEEeqnarray}
Note we have the equivalent conditions \circled{$0$} $\iff$ \circled{$6$} and \circled{$4$} $\iff$ \circled{$7$}.
Particularly, letting $p=1$, we obtain the corresponding sufficient conditions for the deterministic power fluctuation scenario of arbitrary channel side information.
\end{theorem}

\begin{remark}
Particularly, we have
$
\qty( \Tr \qty[e^{\bm{X}}] )^{\vartheta} \nleq  \Tr \qty [ e^{\vartheta \bm{X}} ],\ \exists  \vartheta>0,
$ 
$
\qty( \Tr \qty[{\bm{X}}] )^{\vartheta} \nleq  \Tr \qty [ { \bm{X}}^{\vartheta} ],\ \exists  \vartheta>0,
$
$
e^{\vartheta \Tr[\bm{X}]} \nleq \Tr\qty[ e^{\vartheta\bm{X}} ], \ \exists  \vartheta>0,
$
where $\bm{X}\in\mathfrak{X}_{\ge 0}$. For example, $\bm{X} = \mqty[ 1 & 0\\ 0 & 1]$ and $\vartheta = 2$.
\end{remark}

\begin{theorem}
Consider the frequency-selective MIMO channel with random power fluctuation. If each sub-channel satisfies any one of the sufficient conditions in Theorem \ref{theorem-sufficient-map}, then the distribution of the  overall channel capacity is light-tailed. 
\end{theorem}

\begin{proof}
The proof follows that the light-tail property is preserved for the sum of random variables.
\end{proof}

\begin{remark}
The channel model $\bm{H}\bm{H}^\ast$ is the product formulation of the large-scale fading and small-scale fading effects.
\end{remark}

\begin{remark}
The tail property of the capacity is determined by the product of the random power and the random eigenvalues of the channel matrix.
Thus, it is necessary to investigate the tail property of the product of two random variables.
It is reasonable to assume independence between these two random variables, because the channel side information is not necessarily known at the transmitter.
On the other hand, it is interesting to take into account the dependence for refinement.
\end{remark}

\subsubsection{Tail Distribution of Random Variable Arithmetic}

We study the tail property of the product distribution and sum distribution of random variables, or the impact of the tail property of one random variable on the overall product or sum distribution.

We consider the nonnegative functions, $f(x)$ and $g(x)$, and define the asymptotic notations, $f(x) = O(g(x)) \iff \limsup\limits_{x\rightarrow \infty} \frac{f(x)}{g(x)} < \infty$, $f(x) = \Omega(g(x)) \iff \liminf\limits_{x\rightarrow \infty} \frac{f(x)}{g(x)} > 0$, $f(x) = \Theta(g(x)) \iff f(x) = O(g(x)) \cap f(x) = \Omega(g(x))$, $f(x) = o(g(x)) \iff \lim\limits_{x\rightarrow \infty} \frac{f(x)}{g(x)} = 0$, $f(x) = \omega(g(x)) \iff \lim\limits_{x\rightarrow \infty} \frac{f(x)}{g(x)} = \infty$, and $f(x) \sim g(x) \iff \lim\limits_{x\rightarrow \infty} \frac{f(x)}{g(x)} = 1$.

We define a class of functions $\mathfrak{F}$, $\forall \varphi \in \mathfrak{F}$, $\varphi : \mathbb{R}_{\ge 0}\rightarrow \mathbb{R}_{\ge 0}$, such that $\lim\limits_{x\rightarrow\infty}\varphi(x)=\infty$ and $\lim\limits_{x\rightarrow\infty}\frac{\varphi(x)}{x}=0$, i.e.,
\begin{equation}
\mathfrak{F} = \qty{ \varphi : \mathbb{R}_{\ge 0}\rightarrow \mathbb{R}_{\ge 0}; \lim\limits_{x\rightarrow\infty}\varphi(x)=\infty, \lim\limits_{x\rightarrow\infty}\frac{\varphi(x)}{x}=0 }.
\end{equation}
For example, $\varphi(x) = x^\alpha$, $0<\alpha<1$, or $\varphi(x) = \log(x)$.
This class of functions are useful in decomposing the distribution function of the product or sum
of random variables.

We study the asymptotic behavior of the composition of the function $\mathfrak{F}$ and some classes of tail distributions, e.g., the light-tail distribution, the regularly varying distribution $F \in \mathcal{R}_{\ge 0}$, and the long-tail distribution $F \in \mathcal{L}$ (containing the subexponential distribution as a subset).
We present the proof of the following results in Appendix \ref{proof-lemma-tail-relationships}.

\begin{lemma}\label{lemma-tail-relationships}
Consider the independent random variables $X_i: \Omega \rightarrow \mathbb{R}_{\ge 0}$, $i\in\{ 1, 2 \}$.
\begin{enumerate}[wide, labelwidth=!, labelindent=0pt]
\item
If $F_1 \in \mathcal{L}$, i.e., $\lim\limits_{x\rightarrow \infty} \frac{ \overline{F}_{X_1}(x - y) }{ \overline{F}_{X_1}(x) } =1 $, $\forall y>0$, then, $\overline{F}_{X_1}\qty(x-\varphi(x)) \sim  \overline{F}_{X_1}(x) $ and $\overline{F}_{X_1}\qty( \log \frac{ x }{\varphi(x)}) \sim  \overline{F}_{X_1}(\log x) $, $\forall \varphi \in \mathfrak{F}$.
We specify $\varphi(x) = x^\alpha$, $0<\alpha<1$.
\begin{enumerate}[leftmargin=*]
\item
If 
$\overline{F}_{X_1}(x) = \Omega\qty( x^{-\theta_1} )$, $\theta_1 >0$, and 
$\overline{F}_{X_2}(x) = O\qty(e^{-\theta_2 x}) $, $ \theta_2>0$, then, $\overline{F}_{X_2}\qty(\varphi(x)) = O\qty( \overline{F}_{X_1}(x) )$.

\item
If $\overline{F}_{X_1}(x) = \Omega\qty( x^{-\theta_1} )$, $\theta_1 >0$, and $\overline{F}_{X_2}(x) = O\qty(x^{-\theta_2 }) $, $ \theta_2>0$, and $\alpha \theta_2 > \theta_1$, then, $\overline{F}_{X_2}\qty(\varphi(x)) = O\qty( \overline{F}_{X_1}(x) )$.
\end{enumerate}

\item
If $F_1 \in \mathcal{R}$, i.e., $\overline{F}_{X_1}(x) = L_1(x) x^{-\theta_1}$, where $ \theta_1 \ge 0$, and $\lim\limits_{x\rightarrow \infty}\frac{L_1(tx)}{L_1(x)} = 1$, $\forall t>0$, then, $\overline{F}_{X_1}\qty(x-\varphi(x)) \sim \overline{F}_{X_1}(x)$, $\forall \theta_1 \ge 0$, and $\overline{F}_{X_1}\qty( \frac{ x}{\varphi(x)} )= \omega\qty( \overline{F}_{X_1}(x) ) $, $\forall \theta_1>0$ and $ \overline{F}_{X_1}\qty( \frac{ x}{\varphi(x)} ) \sim \overline{F}_{X_1}(x) $ for $ \theta_1 = 0$, $\forall \varphi \in \mathfrak{F}$.
We specify $\varphi(x) = x^\alpha$, $0<\alpha<1$.
\begin{enumerate}[leftmargin=*]
\item
If $\overline{F}_{X_2}(x) = O\qty(e^{-\theta_2 x}) $, $ \theta_2>0$, then, $\overline{F}_{X_2}\qty(\varphi(x)) = O\qty( \overline{F}_{X_1}(x) )$.
If $\overline{F}_{X_2}(x) = \Theta\qty(e^{-\theta_2 x}) $, $ \theta_2>0$, then, $\overline{F}_{X_2}\qty(\varphi(x)) = o\qty( \overline{F}_{X_1}(x) )$.

\item
If $\overline{F}_{X_2}(x) = O\qty(x^{-\theta_2}) $, $ \theta_2>0$, and $\alpha \theta_2 > \theta_1$, then, $\overline{F}_{X_2}\qty(\varphi(x)) = O\qty( \overline{F}_{X_1}(x) )$.
If $\overline{F}_{X_2}(x) = \Theta\qty(x^{-\theta_2 }) $, $ \theta_2>0$, and $\alpha \theta_2 > \theta_1$, then, $\overline{F}_{X_2}\qty(\varphi(x)) = o\qty( \overline{F}_{X_1}(x) )$.
\end{enumerate}

\item
If $ \overline{F}_{X_1}(x)  = \Theta\qty(x^{-\theta_1})$, $\theta_1 >0$, then, $\overline{F}_{X_1}\qty(x-\varphi(x)) = \Theta\qty( \overline{F}_{X_1}(x) ) $ and $\overline{F}_{X_1}\qty( \frac{ x}{\varphi(x)} )= \omega\qty( \overline{F}_{X_1}(x) ) $, $\forall \varphi \in \mathfrak{F}$.
We specify $\varphi(x) = x^\alpha$, $0<\alpha<1$.
\begin{enumerate}[leftmargin=*]
\item
If $\overline{F}_{X_2}(x) = O\qty(e^{-\theta_2 x}) $, $\theta_2>0$, then, $\overline{F}_{X_2}\qty(\varphi(x)) = O\qty( \overline{F}_{X_1}(x) )$.
If $\overline{F}_{X_2}(x) = \Theta\qty(e^{-\theta_2 x}) $, $\theta_2>0$, then, $\overline{F}_{X_2}\qty(\varphi(x)) = o\qty( \overline{F}_{X_1}(x) )$.

\item
If $\overline{F}_{X_2}(x) = O\qty(x^{-\theta_2 }) $, $\theta_2>0$, and $\alpha \theta_2 \ge \theta_1$, then, $\overline{F}_{X_2}\qty(\varphi(x)) = O\qty( \overline{F}_{X_1}(x) )$.
If $\overline{F}_{X_2}(x) = \Theta\qty(x^{-\theta_2 }) $, $\theta_2>0$, and $\alpha \theta_2 > \theta_1$, then, $\overline{F}_{X_2}\qty(\varphi(x)) = o\qty( \overline{F}_{X_1}(x) )$.
\end{enumerate}

\item
If $ \overline{F}_{X_1}(x)  = \Theta\qty(e^{-\theta_1 x})$, $\theta_1 >0$, then, $\overline{F}_{X_1}\qty( \frac{ x}{\varphi(x)} )= \omega\qty( \overline{F}_{X_1}(x) ) $ and $\overline{F}_{X_1}\qty(x-\varphi(x))= \omega\qty( \overline{F}_{X_1}(x) ) $, $\forall \varphi \in \mathfrak{F}$.
\begin{enumerate}[leftmargin=*]
\item
If $\overline{F}_{X_2}(x) = \Theta\qty(e^{-\theta_2 x}) $, $\theta_2>0$, then, $\overline{F}_{X_2}\qty(\varphi(x)) = \omega\qty( \overline{F}_{X_1}(x) )$, $\forall \varphi \in \mathfrak{F}$.

\item
If $\overline{F}_{X_2}(x) = \Theta\qty(x^{-\theta_2 }) $, $\theta_2>0$, then, $\overline{F}_{X_2}\qty(\varphi(x)) = \omega\qty( \overline{F}_{X_1}(x) )$, $\forall \varphi \in \mathfrak{F}$.
\end{enumerate}
\end{enumerate}
\end{lemma}

\begin{remark}
It is interesting to notice that, from the light-tail to the heavy-tail distributions, the tail behaviors go from $\overline{F}_{X_1}\qty( \frac{ x}{\varphi(x)} )= \omega\qty( \overline{F}_{X_1}(x) ) $ and $\overline{F}_{X_1}\qty( { x} - {\varphi(x)} )= \omega\qty( \overline{F}_{X_1}(x) ) $ to $\overline{F}_{X_1}\qty( \frac{ x}{\varphi(x)} ) \sim \overline{F}_{X_1}(x) $ and $\overline{F}_{X_1}\qty( { x} - {\varphi(x)} ) \sim \overline{F}_{X_1}(x) $, $\forall \varphi \in \mathfrak{F}$.
It is impossible that $\overline{F}_{X_1}\qty( \frac{ x}{\varphi(x)} )= o\qty( \overline{F}_{X_1}(x) ) $ or $\overline{F}_{X_1}\qty( { x} - {\varphi(x)} )= o\qty( \overline{F}_{X_1}(x) ) $, because the complementary cumulative distribution function is non-increasing.
\end{remark}

\begin{remark}
It is known that \cite{embrechts1997modelling} the distribution function $F \in \mathcal{L}$ if and only if $\overline{F}\circ \log \in \mathcal{R}_{0}$, where $(f\circ g)(x) = f(g(x))$.
The distribution function $F$ is regularly varying, i.e., $F\in \mathcal{R}_{>0}$, if and only if, there exists a positive function, $a(t)$, such that \cite{falk2010laws}
\begin{equation}
\lim_{t \rightarrow \infty} \frac{ F(tx) - F(t) }{a(t)} = \frac{1-x^{\alpha}}{\alpha},\ x>0.
\end{equation}
These distribution functions have polynomially decaying tail.
Letting $\alpha \rightarrow 0$, we obtain \cite{falk2010laws}
\begin{equation}
\lim_{t \rightarrow \infty} \frac{ F(tx) - F(t) }{a(t)} = \log(x),\ x>0,
\end{equation}
which characterizes a class of super-heavy distribution functions with slowly varying tails $\mathcal{R}_{0}$.
\end{remark}

\begin{remark}
It is interesting to define and study a new tail behavior, i.e., $\limsup\limits_{x\rightarrow \infty} e^{\epsilon \varphi(x)} \overline{F}(x) = \infty$, $\forall \epsilon>0$, $\exists \varphi \in \mathfrak{F}$.
Note the function $\overline{F}(x)$ is heavy-tailed \cite{foss2013introduction}, if and only if $\limsup\limits_{x\rightarrow \infty} e^{\epsilon x} \overline{F}(x) = \infty$, $\forall \epsilon>0$.
\end{remark}

We present a necessary condition for the product of random variables to be light-tailed.

\begin{theorem}
Consider the independent random variables $X_i: \Omega \rightarrow \mathbb{R}$, $i\in\{ 1, \ldots, N \}$. 
The necessary condition for the existence of the moment generating function of the product of the random variables, $X = \prod_{i=1}^{N} X_i$, is the existence of the means of all the random variables, i.e.,
\begin{multline}
\mathbb{E} \qty[ e^{ \theta X} ] < \infty,\ \exists \theta >0 \\
\implies
\mathbb{E} \qty[ { X_i} ] < \infty,\ \forall i\in\{ 1, \ldots, N \}. 
\end{multline}
\end{theorem}

\begin{proof}
It is easy to show that $\mathbb{E}_{X}\qty[ e^{\theta X} ] = \mathbb{E}_{X_1}\qty[ \mathbb{E}_{X_2} \qty[ \ldots \mathbb{E}_{X_N}\qty[ e^{\theta X} ] ] ] \ge e^{\theta \prod_{i=1}^{N} \mathbb{E}\qty[ X_i ] } $, where the equality follows the independence assumption and the inequality follows the Jensen's inequality. 
\end{proof}

\begin{remark}
The existence of the mean of each random variables is not a sufficient condition for the existence of the moment generating function of the product. For example, there is no moment generating function for the product of standard normal random variables \cite{stojanac2017products}.
\end{remark}

We present a sufficient condition on a random variable, whose product with a fat-tail type random variable remains a fat-tail type random variable.
We present the proof in Appendix \ref{proof-theorem-fat-tail-sufficient}.

\begin{theorem}\label{theorem-fat-tail-sufficient}
Consider the independent random variables $X_i: \Omega \rightarrow \mathbb{R}_{\ge 0}$, $i\in\{ 1, \ldots, N \}$.
Suppose $\overline{F}_{X_1}(x) = \Theta\qty(x^{-\theta})$, 
and $\mathbb{E}\qty[ X_{j}^\theta ] <\infty$, $\forall 2\le j\le N$, $\exists\theta>0$.
Then, we have 
\begin{equation}
\overline{F}_{\prod_{i=1}^{N}X_i}(x) = \Theta\qty(x^{-\theta}).
\end{equation}
\end{theorem}

\begin{remark}
If there exists one and only one random variable, $\mathbb{E}\qty[ X_j^\theta ] = 0$, $ j\in \{2, 3,\ldots, N \}$, $\forall \theta>0$, then 
$
\overline{F}_{\prod_{i=1}^{N}X_i}(x) = o\qty(x^{-\theta}).
$
Letting this random variable be the last one for multiplication yields the proof.
\end{remark}

\begin{remark}
Since $f(x)\sim g(x) \implies f(x)= \Theta(g(x))$,
if $\overline{F}_{X_1}(x) \sim C_1 x^{-\theta}$, $\exists C_1>0$,
and $\mathbb{E}\qty[ X_{j}^\theta ] <\infty$, $\forall 2\le j\le N$, $\exists\theta>0$,
then, 
$
\overline{F}_{\prod_{i=1}^{N}X_i}(x) = \Theta\qty(x^{-\theta}).
$
\end{remark}

\begin{remark}
If $\overline{F}_{X_1}(x) \sim C_1 x^{-\theta}$, $\exists C_1>0$,
and $\mathbb{E}\qty[ X_{j}^\theta ] <\infty$, $\forall 2\le j\le N$, $\exists\theta>0$,
then, 
\begin{equation}
\overline{F}_{\prod_{i=1}^{N}X_i}(x) \sim \prod_{j=2}^{N} \mathbb{E}\qty[ X_{j}^\theta ] \cdot C_1 x^{-\theta}.
\end{equation}
The proof of the case $N=2$ is available in \cite{buraczewski2016stochastic} and the proof of the general case follows by iteration.
\end{remark}

We present a sufficient condition on a random variable, whose product with a fat-tail upper bounded random variable remains a fat-tail upper bounded random variable.
We present the proof in Appendix \ref{proof-theorem-fat-product-sufficient}.

\begin{theorem}\label{theorem-fat-product-sufficient}
Consider the independent random variables $X_i: \Omega \rightarrow \mathbb{R}_{\ge 0}$, $i\in\{ 1, \ldots, N \}$.
Suppose $\overline{F}_{X_1}(x) = O\qty(x^{-\theta})$, 
$\exists \varphi_j \in \mathfrak{F}$,
$\overline{F}_{X_j}\qty(\varphi_j(x)) = O\qty(x^{-\theta})$, and $\mathbb{E}\qty[ X_{j}^\theta ] <\infty$, $\forall 2\le j\le N$, $\exists\theta>0$.
Then, we have 
\begin{equation}
\overline{F}_{\prod_{i=1}^{N}X_i}(x) = O\qty(x^{-\theta}).
\end{equation}
\end{theorem}

\begin{remark}
Since $\lim_{x\rightarrow\infty}\frac{\varphi_2(x)}{x}=0$, we have, $\exists x_0>0$, $\forall x > x_0$, $x\ge \varphi_2(x)$ and 
$ \overline{F}_{X_2}\qty( x ) \le \overline{F}_{X_2}\qty( \varphi_2(x) )$.
Thus, if $\overline{F}_{X_2}\qty(\varphi_2(x)) = O\qty(x^{-\theta})$, where $\lim_{x\rightarrow \infty} \varphi_2(x)=\infty$, then $\overline{F}_{X_2}\qty( x ) = O\qty(x^{-\theta})$.
\end{remark}

\begin{remark}
A related result concerning the sharp approximation, $\overline{F}(x) \sim C x^{-\theta}$, where $\theta>0$ and $C>0$ is a constant, is available in \cite{sarantsev2011tail}.
Another reference is \cite{jessen2006regularly}.
\end{remark}

\begin{remark}
The logarithm transform expression of the capacity formula indicates that it is sufficient to study the $f(x) = O(g(x))$ approximation rather than the sharper $f(x) \sim g(x)$ approximation, because a light-tailed distribution of capacity can bear as heavy as fat-tailed distributions of the random power and the random fade in the logarithm function. 
\end{remark}

We present the results of the tail upper bounds of the product and sum of random variables.
We present the proof in Appendix \ref{proof-theorem-super-heavy-tail-bound}.

\begin{theorem}\label{theorem-super-heavy-tail-bound}
Consider the independent random variables $X_i: \Omega \rightarrow \mathbb{R}_{\ge 0}$, $i\in\{ 1, \ldots, N \}$.
\begin{enumerate}[wide, labelwidth=!, labelindent=0pt]
\item
Suppose, $\exists \varphi_j \in \mathfrak{F}$,
$\overline{F}_{X_j}\qty(\varphi_j(x)) = O\qty( \overline{F}_{X_1}(x) )$, 
and $\overline{F}_{X_1} \qty(\frac{x}{\varphi_j(x)}) = O\qty( \overline{F}_{X_1}(x) )$,
$\forall 2\le j\le N$.
Then, we have 
\begin{equation}
\overline{F}_{\prod_{i=1}^{N}X_i}(x) = O\qty( \overline{F}_{X_1}(x) ).
\end{equation}

\item
Suppose, $\exists \varphi_j \in \mathfrak{F}$,
$\overline{F}_{X_j}\qty(\varphi_j(x)) = O\qty( \overline{F}_{X_1}(x) )$, and $\overline{F}_{X_1}\qty(x-\varphi_j(x))= O\qty( \overline{F}_{X_1}(x) ) $, $\forall 2\le j\le N$.
Then, we have
\begin{equation}
\overline{F}_{\sum_{i=1}^{N} X_i} (x) = O \qty( \overline{F}_{X_1}(x) ).
\end{equation}
\end{enumerate}
\end{theorem}

\begin{remark}
Since $\lim_{x\rightarrow\infty}\frac{\varphi_2(x)}{x}=0$, we have, $\exists x_0>0$, $\forall x > x_0$, $x\ge \varphi_2(x)$ and 
$ \overline{F}_{X_2}\qty( x ) \le \overline{F}_{X_2}\qty( \varphi_2(x) )$.
Thus, if $\overline{F}_{X_2}\qty(\varphi_2(x)) = O\qty( \overline{F}_{X_1}(x) )$, where $\lim_{x\rightarrow \infty} \varphi_2(x)=\infty$, then $\overline{F}_{X_2}\qty( x ) = O\qty( \overline{F}_{X_1}(x) )$.
\end{remark}

We present a further result, which indicates that the tail behavior of the random variables product can not be effectively transformed from heavy to light, by the product or sum with other random variables, when there is a slowly varying distribution in the product or a regularly varying distribution in the sum.
We present the proof in Appendix \ref{proof-theorem-slowly-regularly-varying-dominate}.

\begin{theorem}\label{theorem-slowly-regularly-varying-dominate}
Consider the independent random variables $X_i: \Omega \rightarrow \mathbb{R}_{\ge 0}$, $i\in\{ 1, \ldots, N \}$.
\begin{enumerate}[wide, labelwidth=!, labelindent=0pt]
\item
Suppose, $\exists \varphi_j \in \mathfrak{F}$,
$\overline{F}_{X_j}\qty(\varphi_j(x)) = o\qty( \overline{F}_{X_1}(x) )$, 
$F_1 \in \mathcal{R}_0$, i.e., $\overline{F}_{X_1}(t x) \sim \overline{F}_{X_1}(x)$, $\forall t>0$,
$\forall 2\le j\le N$.
Then, we have 
\begin{equation}
\overline{F}_{\prod_{i=1}^{N}X_i}(x) \sim \overline{F}_{X_1}(x).
\end{equation}

\item
Suppose, $\exists \varphi_j \in \mathfrak{F}$,
$\overline{F}_{X_j}\qty(\varphi_j(x)) = o\qty( \overline{F}_{X_1}(x) )$, and $F_{X_1} \in \mathcal{L}$, i.e., $\overline{F}_{X_1}\qty(x-t) \sim \overline{F}_{X_1}(x)$, $\forall t>0$,
$\forall 2\le j\le N$.
Then, we have
\begin{equation}
\overline{F}_{\sum_{i=1}^{N} X_i} (x) \sim \overline{F}_{X_1}(x).
\end{equation}
\end{enumerate}
\end{theorem}

\begin{remark}
According to Lemma \ref{lemma-tail-relationships}, the situation for the product of random variables appears for slowly varying distributions $F_1\in \mathcal{R}_0$, and $F_j (x) = \Theta\qty(x^{-\theta_j})$ or $F_j (x) = \Theta\qty(e^{-\theta_j x})$, $\theta_j >0$, $j\in\{ 2, \ldots, N \}$;
and the situation for the sum of random variables appears for regularly varying distributions $F_1\in \mathcal{R}_{\ge 0}$, and $F_j (x) = \Theta\qty(x^{-\theta_j})$ or $F_j (x) = \Theta\qty(e^{-\theta_j x})$, $\theta_j >0$, $j\in\{ 2, \ldots, N \}$.
\end{remark}

\begin{remark}
It is interesting to study the possibility of transforming the tail heaviness of a random variable from heavy to light through some functions with other random variables, e.g., for some special cases.
\end{remark}

\begin{remark}
The asymptotic behavior of the right tail of the sum of the random variables can be insensible to both positive and negative dependence \cite{albrecher2006tail}\cite{tang2008insensitivity}\cite{ko2008sums}\cite{asmussen2008asymptotics}\cite{geluk2009asymptotic}, while the asymptotic behavior of the left tail can be connected with the dependence structures \cite{gulisashvili2016tail}\cite{tankov2016tails}. 
In addition, there are scenarios, where the right tail of the sum distribution is sensitive to the dependence structures \cite{albrecher2006tail}.
\end{remark}

\begin{remark}
The tail behavior of the product distribution is more complicated.
For example, it is shown that the product distribution of two independent random variables with exponential distributions is subexponential \cite{tang2008light}\cite{liu2010subexponential}.
Particularly, the dependence between the random variables are crucial for the tail behavior of the product distribution \cite{jiang2011product}\cite{hashorva2014tail}\cite{chen2018extensions}, e.g., the dependence can either decrease or increase the product distribution tail heaviness compared to the independence scenario \cite{jiang2011product}\cite{yang2013subexponentiality}.
In addition, the tail of the product distribution with dependence can be asymptotically bounded above and below by the tail of a dominating random variable \cite{yang2011tail}\cite{chen2018extensions} or can be asymptotically bounded above and below by the tail with assumption of independence \cite{jiang2011product}\cite{yang2013subexponentiality}. 
\end{remark}

\begin{remark}
It is interesting to investigate the extreme influence of the dependence among the random parameters in the wireless channel capacity on the tail behavior of the marginal distribution of the capacity, e.g., whether or not the dependence between two light-tailed or heavy-tailed random variables can cause a super-heavy tail of the product or sum distribution.
For example, considering the comonotonic random variables with identical distributions \cite{dhaene2002concept}, $X_i \sim X$, $1\le i\le N$, we have $\overline{F}_{\sum\limits_{1\le i\le N} X_i}(x) = \overline{F}_{X}\qty(x/N)$ and $\overline{F}_{\prod\limits_{1\le i\le N} X_i}(x) = \overline{F}_{X}\qty(x^{1/N})$, compared to the distribution $\overline{F}_{X}\qty(x)$, the sum distribution $\overline{F}_{X}\qty(x/N)$ is scale invariant for Pareto Type I distribution and asymptotically scale invariant for regular varying distributions, and the product distribution $\overline{F}_{X}\qty(x^{1/N})$ has a smaller tail index for Pareto Type I distribution.
\end{remark}

\section{Dependence Transform}
\label{transformability}

We provide the dependence manipulation techniques for both the spatial dependence and temporal dependence of a stochastic process.
The manipulation of the spatial dependence means the dependence manipulation of the random parameters of the stochastic process at some time epochs, while the manipulation of the temporal dependence means the manipulation of some random parameters on the time line.
Primary results of temporal dependence manipulation are shown in \cite{sun2018hidden}.

We define the unconditionally increasingly functions
\begin{multline}
\mathfrak{F}_{UI} = \{ f: \mathbb{R}^n \rightarrow \mathbb{R}; f(x_i| \qty{\bm{x}\setminus x_i }) \text{ is increasing at } x_i, \\
\forall \bm{x} \in \mathbb{R}^{n}, \forall 1\le i\le n \},
\end{multline}
and if the function $f$ is strictly increasing, we denote $f \in \mathfrak{F}_{USI}$.
We define the unconditionally increasingly affine functions
\begin{multline}
\mathfrak{F}_{UIA} = \{ f: \mathbb{R}^n \rightarrow \mathbb{R}; 
f(x_i| \qty{\bm{x}\setminus x_i }) \text{ is affine function of } x_i, \\
\text{ and is increasing at } x_i, 
\forall \bm{x} \in \mathbb{R}^{n}, \forall 1\le i\le n \},
\end{multline}
and if the function $f$ is strictly increasing, we denote $f \in \mathfrak{F}_{USIA}$.
These function classes are sufficiently general in some scenarios, e.g., the wireless channel capacity process.

\subsection{Identical Marginals}

We present a sufficient condition that the composition of a supermodular function with some multivariate functions is a supermodular function.

\begin{lemma}\label{lemma-sm-function}
Let $f_t: \mathbb{R}^n \rightarrow \mathbb{R}$, $\forall t\ge 1$.
If $g: \mathbb{R}^t \rightarrow \mathbb{R}$ is supermodular, $g^\prime := g \qty(f_1, \ldots, f_t) : \mathbb{R}^{t\times n} \rightarrow \mathbb{R}$, and
\begin{multline}
g\qty( \qty( f_1(\bm{x}_1), \ldots, f_t(\bm{x}_t) ) \wedge \qty( f_1(\bm{y}_1), \ldots, f_t(\bm{y}_t) ) ) \\
+ g\qty( \qty( f_1(\bm{x}_1), \ldots, f_t(\bm{x}_t) ) \vee \qty( f_1(\bm{y}_1), \ldots, f_t(\bm{y}_t) ) )\\
= 
g^\prime \qty( (\bm{x}_1, \ldots, \bm{x}_t) \wedge (\bm{y}_1, \ldots, \bm{y}_t) ) \\
+ g^\prime \qty( (\bm{x}_1, \ldots, \bm{x}_t) \vee (\bm{y}_1, \ldots, \bm{y}_t) ),
\end{multline}
then $g^\prime$ is supermodular.
\end{lemma}

\begin{proof}
Considering 
\begin{multline}
g\qty( \qty( f_1(\bm{x}_1), \ldots, f_t(\bm{x}_t) ) \wedge \qty( f_1(\bm{y}_1), \ldots, f_t(\bm{y}_t) ) ) \\
+ g\qty( \qty( f_1(\bm{x}_1), \ldots, f_t(\bm{x}_t) ) \vee \qty( f_1(\bm{y}_1), \ldots, f_t(\bm{y}_t) ) ) \\
\ge g^\prime\qty( \bm{x}_1, \ldots, \bm{x}_t ) +  g^\prime\qty( \bm{y}_1, \ldots, \bm{y}_t ),
\end{multline}
the proof follows directly.
\end{proof}

\begin{remark}
The implicit function theorem \cite{spivak1965calculus}\cite{rudin1976principles} gives a sufficient condition on the functions for the existence of their inverse in general.
\end{remark}

We investigate the scenario, where the multidimensional process is temporally independent and spatially dependent. 

\begin{theorem}\label{theorem-sm-sp-dependent}
Assume the random parameters are spatially dependent and temporally independent.
Assume $f_t : \mathbb{R}^n \rightarrow \mathbb{R}$, 
\begin{equation}
f_t\qty(\bm{x}_t) \diamond f_t\qty(\bm{y}_t) = f_t\qty( \bm{x}_t \hat{\diamond} \bm{y}_t ),\ \forall \bm{x}_t, \bm{y}_t \in \mathbb{R}^{n},\ \forall t\ge 0,
\end{equation}
where $\diamond, \hat{\diamond} \in \qty{\wedge, \vee}$ preserves one of the following three relations for all $t\ge 0$$:$ $\diamond = \hat{\diamond}$, $\qty{\diamond = \wedge|\hat{\diamond}=\vee} $, and $\qty{\diamond = \vee|\hat{\diamond}=\wedge}$.

If, for any $1\le j\le t,$
\begin{equation}
\left( X_j^1, X_j^2, \ldots, X_j^n \right) \le_{sm} \left( \widetilde{X}_j^1, \widetilde{X}_j^2, \ldots, \widetilde{X}_j^n \right),
\end{equation}
and $\qty( X_j^1, \ldots, X_j^n )$ and $\qty( \widetilde{X}_j^1, \ldots, \widetilde{X}_j^n )$, and $\qty( X_k^1, \ldots, X_k^n )$ are independent for all $j\neq k$, 
then
\begin{equation}
 \left( X_1, X_2, \ldots, X_t \right) \le_{sm} \left( {{X}}_1, {{X}}_2,\ldots, \widetilde{X}_j, \ldots, {{X}}_t \right),
\end{equation}
where ${{X}}_i = f_i \qty( {X}_i^1, \ldots,  {X}_i^n )$, $\forall 1\le i\le t$, and ${\widetilde{X}}_j = f_j \qty( \widetilde{X}_j^1, \ldots,  \widetilde{X}_j^n )$, $1\le j\le t$.

If
\begin{equation}
\left( X_j^1, X_j^2, \ldots, X_j^n \right) \le_{sm} \left( \widetilde{X}_j^1, \widetilde{X}_j^2, \ldots, \widetilde{X}_j^n \right),\ \forall 1\le j \le t,
\end{equation}
and $\qty( X_j^1, \ldots, X_j^n )$ and $\qty( X_k^1, \ldots, X_k^n )$ are independent for all $j\neq k$, so are $\qty( \widetilde{X}_j^1, \ldots, \widetilde{X}_j^n )$ and $\qty( \widetilde{X}_k^1, \ldots, \widetilde{X}_k^n )$,
then
\begin{multline}
 \left( \widetilde{X}_1,\ldots, \widetilde{X}_k, {X}_{k+1}, \ldots, X_t \right) \\
 \le_{sm} \left( {\widetilde{X}}_1, \ldots, {\widetilde{X}}_j, X_{j+1}, \ldots, {{X}}_t \right),\ \forall 1\le k\le j\le t,
\end{multline}
where ${{X}}_j = f_j \qty( {X}_j^1, \ldots,  {X}_j^n )$ and ${\widetilde{X}}_j = f_j \qty( \widetilde{X}_j^1, \ldots,  \widetilde{X}_j^n )$, $\forall 1\le j\le t$.
\end{theorem}

\begin{proof}
Considering the temporal independence assumption and the conjunction property of supermodular order \cite{shaked2007stochastic},
we have
$
\qty( X_1^1, \ldots, X_1^n, \ldots, X_t^1, \ldots, X_t^n )
\le_{sm}
\qty( {X}_1^1, \ldots, {X}_1^n, \ldots, \widetilde{X}_j^1, \ldots, \widetilde{X}_j^n, \ldots, {X}_t^1, \ldots, {X}_t^n )
$.

Letting $g: \mathbb{R}^t \rightarrow \mathbb{R}$ be supermodular and denote $g^\prime := g \qty(f_1, \ldots, f_t) : \mathbb{R}^{t\times n} \rightarrow \mathbb{R}$,
we have $g^\prime$ is supermodular, which follows Lemma \ref{lemma-sm-function}. 
Thus, it directly implies
$
\qty( f_1 \qty( X_1^1, \ldots, X_1^n ), \ldots, f_t\qty( X_t^1, \ldots, X_t^n ) )
\le_{sm}
\qty( f_1 \qty({X}_1^1, \ldots, {X}_1^n ), \ldots, f_j \qty( \widetilde{X}_j^1, \ldots, \widetilde{X}_j^n ), \ldots, f_t \qty( {X}_t^1, \ldots, {X}_t^n ) )
$.

The proof of the other result follows the reflexivity and transitivity property of supermodular order \cite{muller2002comparison}.
\end{proof}

\begin{remark}
The results indicate that the spatial dependence of the random parameters also influences the dependence of the stochastic process and more manipulations of the spatial dependence has more strength to transform the dependence of the stochastic process.
\end{remark}

\begin{remark}
It is interesting to investigate the relationship between the requirement of the functional and the spatial dependence of the random parameters.
An example is the comonotonicity dependence structure with identical marginal distribution, i.e., the random parameters are equal almost surely, thus the requirement of the function reduces to the scenario of the requirement of the univariate functional scenario, i.e., decreasing or increasing for each variate on the function domain.
\end{remark}

We present a result without specification on the spatial and temporal dependence. Note the relaxation of the specification on dependence is replaced by the additional conditions on the functionals.

\begin{theorem}\label{theorem-sm-dependence}
Assume $f_t : \mathbb{R}^n \rightarrow \mathbb{R}$ and $f_t\qty(\bm{X}_t^i | \bm{Z}_t^i = \bm{z}_t^i )$ are all increasing or all decreasing at each component of $\bm{X}^i = \qty(X^i_1, \ldots, X^i_t)$, for any $\bm{z}_t^i = \qty(x_t^1,\ldots, x_t^{i-1}, x_t^{i+1}, \ldots, x_t^n)$ in the support of $\bm{Z}_t^i = \qty(X^1_t,\ldots, X_t^{i-1}, X_t^{i+1}, \ldots, X_t^n) $, $\forall 1\le i\le n$, $\forall t\ge 1$.
Denote $\bm{z}^i = \qty{ \bm{z}_1^i, \ldots, \bm{z}_t^i }$, $\bm{Z}^i = \qty{ \bm{Z}_1^i, \ldots, \bm{Z}_t^i }$, and $\qty(\bm{X}^i | \bm{Z}^i) \le_{sm} \qty(\widetilde{\bm{X}}^i | \bm{Z}^i) \iff \qty(\bm{X}^i | \bm{Z}^i = \bm{z}^i ) \le_{sm} \qty(\widetilde{\bm{X}}^i | \bm{Z}^i =\bm{z}^i ),\ \forall \bm{z}^i \in \bm{Z}^i$.

If, for any $1\le i \le n,$
\begin{equation}
\left( X_1^i, X_2^i, \ldots, X_t^i | \bm{Z}^i \right) 
\le_{sm} \left( \widetilde{X}_1^i, \widetilde{X}_2^i, \ldots, \widetilde{X}_t^i | \bm{Z}^i \right), 
\end{equation}
then
\begin{equation}
 \left( X_1, X_2, \ldots, X_t \right) \le_{sm} \left( {\widetilde{X}}_1, {\widetilde{X}}_2, \ldots, {\widetilde{X}}_t \right),
\end{equation}
where ${\widetilde{X}}_j = f_j ( X_j^1, \ldots, X_j^{i-1}, \widetilde{X}_j^i, X_j^{i+1}, \ldots, X_j^n )$, $\forall 1\le j\le t$.

If, $\forall 1\le j \le i$,
\begin{equation}
\left( X_1^j, X_2^j, \ldots, X_t^j | \bm{Z}^j \right) 
\le_{sm} \left( \widetilde{X}_1^j, \widetilde{X}_2^j, \ldots, \widetilde{X}_t^j | \bm{Z}^j \right),
\end{equation}
then
\begin{equation}
\widetilde{\bm{X}}_{t}^{k} \le_{sm} \widetilde{\bm{X}}_{t}^{j},\ \forall 0\le k \le j \le i, 
\end{equation}
with $\widetilde{{X}}_{{t}_m}^{l} = f_m ( \widetilde{X}_m^1, \ldots, \widetilde{X}_m^l, {X}_m^{l+1}, \ldots, {X}_m^n )$, $\ 1\le m \le t$, and $\widetilde{\bm{X}}_{t}^{l} = ( \widetilde{{X}}_{t_1}^{l}, \ldots, \widetilde{{X}}_{t_t}^{l} )$, $l\in\{k,j\}$.
\end{theorem}

\begin{proof}
The proof follows analogically to the proof of the independence scenario \cite{sun2018hidden}, by using the conditional probability.
\end{proof}

\begin{remark}
The stochastic orders $\qty(\bm{X}|\bm{Z} = \bm{z}) \le_{sm} \qty(\bm{Y}|\bm{Z} = \bm{z})$ and $\mathbb{E}\qty(\bm{X}|\bm{Z}) \le_{sm} \mathbb{E}\qty(\bm{Y}|\bm{Z})$ correspond to the conditional supermodular order in the sense of the uniform conditional ordering
\cite{whitt1980uniform}\cite{ruschendorf1991conditional}.
On the one hand, the conditional formulation influences the stochastic ordering of the probability measures, moreover, it influences the property of the functions of the random variables, e.g., the monotonicity.
\end{remark}

\begin{remark}
There is an implicit condition that the spatial dependence must not influence the temporal dependence ordering, or the temporal dependence ordering is conditional on the spatial dependence.
Specifically, if spatial independence is assumed, the conditional event disappears.
On the other hand, it is interesting to investigate what type of spatial dependence sufficiently imply the conditional ordering.
\end{remark}

\begin{remark}
It is interesting to investigate the conditional probability and conditional stochastic order expression of the spatial dependence manipulation scenario, e.g., Theorem \ref{theorem-sm-sp-dependent}.
\end{remark}

\begin{remark}
As an example of conditional probability, the function of two random variables $f(X,Y)$,
the independence assumption implies $\mathbb{E}\qty[f(X,Y)] = \mathbb{E}_{X}\mathbb{E}_{Y}\qty[f(X,Y)]$, while the absence of independence implies that $\mathbb{E}\qty[f(X,Y)] = \mathbb{E}_{X}\mathbb{E}_{Y|X=x}\qty[f(X,Y)| X=x]$.
\end{remark}

\begin{remark}
An an example of conditional monotonicity of functions, let $f(x,y) = x^y$, $x>0$, then $f(x,y|y>0)$ is increasing at $x$, $f(x,y|y<0)$ is decreasing at $x$, and $f(x,y|y=0)=1$ is constant.
\end{remark}

\begin{remark}
It is interesting to extend the results in Theorem \ref{theorem-sm-dependence} and the corresponding results without conditional probability to the increasing supermodular order $\le_{ism}$ and the symmetric supermodular order $\le_{symsm}$.
\end{remark}

\subsection{Different Marginals}

We present a result about the manipulation of stochastic process based on the marginals.

\begin{theorem}\label{theorem-marginal-copula}
Assume the random parameters are spatially independent and temporally dependent.
If, for any $1\le i\le n$, 
$
\left( X_1^i, X_2^i, \ldots, X_t^i \right) \le_{dcx} \left( \widetilde{X}_1^i, \widetilde{X}_2^i, \ldots, \widetilde{X}_t^i \right), 
$
$f_j\qty( X_j^1\ldots, X_j^n ) \in \mathfrak{F}_{USIA}$, $\forall 1\le j\le t$, 
and $\left( X_1^i, X_2^i, \ldots, X_t^i \right)$ and $\left( {\widetilde{X}}_1^i, {\widetilde{X}}_2^i, \ldots, {\widetilde{X}}_t^i \right)$ have the common conditionally increasing copula 
$C^i_t \qty(u_1^i,\ldots, u_t^i)$,
then
\begin{equation}
 \left( X_1, X_2, \ldots, X_t \right) \le_{dcx} \left( {\widetilde{X}}_1, {\widetilde{X}}_2, \ldots, {\widetilde{X}}_t \right),
\end{equation}
where ${\widetilde{X}}_j = f_j ( X_j^1, \ldots, X_j^{i-1}, \widetilde{X}_j^i, X_j^{i+1}, \ldots, X_j^n )$, $\forall 1\le j\le t$,
and 
\begin{equation}
\sum_{j=1}^{t} \alpha_j X_j \le_{cx} \sum_{j=1}^{t} \alpha_j \widetilde{X}_j.
\end{equation}
where $\alpha_j \in \mathbb{R}_{\ge 0}$, $\forall 1\le j\le t$.
\end{theorem}

\begin{proof}
Without loss of generality, we consider the first variate.
We have 
$\qty(X_1^1, \ldots, X_t^1) \le_{dcx} \qty(\widetilde{X}_1^1, \ldots, \widetilde{X}_t^1) 
\implies X_j^1 \le_{cx} \widetilde{X}_j^1,\ \forall 1\le j\le t  
$.
Considering that the composition $g\circ f$ of a convex function $g$ and an affine function $f$ is a convex function \cite{simchi2005logic}, we have 
$
f_j \qty( X_j^1, x_j^2, \ldots, x_j^n ) \le_{cx} f_j \qty( \widetilde{X}_j^1, {x}_j^2, \ldots, x_j^n ), \ \forall \qty(x_j^2, \ldots, x_j^n) \in \qty(X_j^2,\ldots, X_j^n),\ \forall 1\le j\le t
$.

Since the functional $\qty{f_j}_{1\le j\le t}$ is unconditionally increasing, the copula of the sequence, $\qty{ f_j \qty( X_j^1, x_j^2, \ldots, x_j^n ) }_{1\le j\le t}$, equals the copula of the sequence $\qty{ X_j^1}_{1\le j\le t}$.
Thus, we obtain $\qty( f_1\qty( X_1^1, x_1^2, \ldots, x_1^n ), \ldots, f_t\qty( X_t^1, x_t^2, \ldots, x_t^n ) ) \le_{dcx} \qty( f_1\qty( \widetilde{X}_1^1, x_1^2, \ldots, x_1^n ), \ldots, f_t\qty( \widetilde{X}_t^1, x_t^2, \ldots, x_t^n ) )$, $\forall \qty(x_j^2, \ldots, x_j^n) \in \qty(X_j^2,\ldots, X_j^n)$, $\forall 1\le j\le t$.
This conditional case follows the proof of the one dimensional result in \cite{muller2001stochastic}.
By taking the expectation, we obtain
$ \left( X_1, X_2, \ldots, X_t \right) \le_{dcx} \left( {\widetilde{X}}_1, {\widetilde{X}}_2, \ldots, {\widetilde{X}}_t \right) $, which further implies the convex order of the weighted sum \cite{muller2001stochastic}.
\end{proof}

\begin{remark}
It is interesting to consider the copula construction of the functional sequence based on the copulas of each sub-sequence, especially for some special cases, e.g., without Granger causality .
\end{remark}

\begin{remark}
The invariance of copula under strictly increasing transformation of random variables requires that the functionals are strictly increasing.
\end{remark}

\begin{remark}
This result is extensible to any functions $f_j\qty( X_j^1\ldots, X_j^n ) \in \mathfrak{F}$, $\forall 1\le j\le t$, which are componentwisely and strictly increasing and preserves the convexity under the composition $g\circ f_j\qty(X_j^i| \qty{\bm{x}_j \setminus x_j^i} )$, $\forall 1\le i \le n$, with a convex function $g: \mathbb{R} \rightarrow \mathbb{R}$.
\end{remark}

\begin{remark}
The manipulation of marginal distributions is more complicated in the sense that there is a more involved requirement on the functionals.
\end{remark}

We present the result of the strength of the marginal distribution manipulation.

\begin{theorem}\label{theorem-dcx-manipulation-strength}
Assume the random parameters are spatially independent and temporally dependent.
If, for all $1\le i\le n$, 
$\left( X_1^i, X_2^i, \ldots, X_t^i \right) \le_{dcx} \left( \widetilde{X}_1^i, \widetilde{X}_2^i, \ldots, \widetilde{X}_t^i \right)$,
$f_j\qty( X_j^1\ldots, X_j^n ) \in \mathfrak{F}_{USIA}$, $\forall 1\le j\le t$, 
and $\left( X_1^i, X_2^i, \ldots, X_t^i \right)$ and $\left( {\widetilde{X}}_1^i, {\widetilde{X}}_2^i, \ldots, {\widetilde{X}}_t^i \right)$ have the common conditionally increasing copula 
$C^i_t \qty(u_1^i,\ldots, u_t^i)$,
then
\begin{equation}
\widetilde{\bm{X}}_{t}^{k} \le_{dcx} \widetilde{\bm{X}}_{t}^{k^\prime},\ \forall 0\le k \le k^\prime \le n, 
\end{equation}
with $\widetilde{{X}}_{{t}_m}^{l} = f_m ( \widetilde{X}_m^1, \ldots, \widetilde{X}_m^l, {X}_m^{l+1}, \ldots, {X}_m^n )$, $\ 1\le m \le t$, and $\widetilde{\bm{X}}_{t}^{l} = ( \widetilde{{X}}_{t_1}^{l}, \ldots, \widetilde{{X}}_{t_t}^{l} )$, $l\in\{k,k^\prime \}$, 
and 
\begin{equation}
\sum_{j=1}^{t} \alpha_j X_j^k \le_{cx} \sum_{j=1}^{t} \alpha_j \widetilde{X}_j^{k^\prime}.
\end{equation}
where $\alpha_j \in \mathbb{R}_{\ge 0}$, $\forall 1\le j\le t$.
\end{theorem}

\begin{proof}
The proof follows the proof of Theorem \ref{theorem-marginal-copula} and the transitivity of the directionally convex order.
\end{proof}

\begin{remark}
The results show that the manipulation of the marginal distributions of one dimension is able to transform the distribution ordering properties of the overall stochastic process, the more dimensions the more manipulation strength.
In addition, the marginal distribution manipulation is feasible only for positive dependence, while the dependence structure manipulation has no dependence bias.
\end{remark}

\begin{remark}
The results have a dependence control utility.
On the one hand, it means that the negative dependence endows the advantage of reducing the power or capacity cost while attaining a higher performance, which is a physical perspective, on the other hand, it means that the advantage of negative dependence maps to the property of the negative dependence and the convex order, which is a mathematical perspective.
Thus, the mathematical property corresponds to a physical resource, which can be taken advantage of and exploited.
In addition, it indicates that the marginal distribution manipulation is invalid for negative dependence and should be avoided in practice.
\end{remark}

\begin{corollary}
The results in Theorem \ref{theorem-marginal-copula} and Theorem \ref{theorem-dcx-manipulation-strength} extend to the increasing directionally convex order $\le_{idcx}$ of the random vectors and the corresponding increasing convex order $\le_{icx}$ of the weighted partial sums.
\end{corollary}

\begin{proof}
The proof follows the proofs of Theorem \ref{theorem-marginal-copula} and Theorem \ref{theorem-dcx-manipulation-strength}, and the result that the composition of an increasing convex function and an increasing affine function is an increasing convex function \cite{shaked2007stochastic}, and the result that \cite{balakrishnan2012increasing}, letting $\bm{X}$ and $\bm{Y}$ be random vectors with a common conditionally increasing copula and assuming that $X_i \le_{icx} Y_i$, $\forall i$, then $\bm{X} \le_{idcx} \bm{Y}$.
\end{proof}

\begin{corollary}
Assume the random parameters are spatially dependent and temporally independent.
If $f_j\qty( X_j^1\ldots, X_j^n )$, $\forall 1\le j\le t$, are increasing and directionally convex, 
and
$
\left( X_j^1, X_j^2, \ldots, X_j^n \right) \le_{idcx} \left( \widetilde{X}_j^1, \widetilde{X}_j^2, \ldots, \widetilde{X}_j^n \right), 
$ 
$\forall 1\le j\le t$, 
then the same results hold as in Theorem \ref{theorem-sm-sp-dependent}, but in the sense of the increasing and directionally convex order $\le_{idcx}$.
\end{corollary}

\begin{proof}
Note the composition of an increasing and convex function $g: \mathbb{R} \rightarrow \mathbb{R}$ and an increasing and directionally convex function $f: \mathbb{R}^n \rightarrow \mathbb{R}$ is an increasing and directionally convex function $g\circ f$ \cite{shaked2007stochastic}.
Then, the proof follows that the increasing and convex order of each elements implies the increasing and directionally convex order of the random vector with independent elements.
\end{proof}

\begin{remark}
It is interesting to extend the temporal and spatial manipulation results to the conditional (increasing) directionally convex order.
\end{remark}

Assuming independence among the random vectors, we present the directionally convex order result for random sums, which are not necessarily independent.

\begin{theorem}\label{theorem-random-multiplexing-dcx}
Let $\bm{X}_j = ( X_{j,1}, \ldots, X_{j,m} )$ and $\bm{Y}_j = ( Y_{j,1}, \ldots, Y_{j,m} )$, $j = 1,2,\ldots$, be two sequences of non-negative random vectors with independence among components, and let $\bm{M} = \left(M_1, M_2, \ldots , M_m \right)$ and $\bm{N} = \left( N_1, N_2, \ldots, N_m \right)$ be two vectors of non-negative integer-valued random variables. Assume that both $\bm{M}$ and $\bm{N}$ are independent of the $\bm{X}_j$'s and $\bm{Y}_j$'s. 
Assume that $X_{j,i} \le_{cx} X_{j+1,i}$, $\forall 1\le i \le m$, $\forall j \ge 1$.

If $\bm{M} \le_{dcx} \bm{N}$, then
\begin{equation}
\left( \sum_{j=1}^{M_1} X_{j,1}, \ldots, \sum_{j=1}^{M_m} X_{j,m}  \right) 
\le_{dcx} \left( \sum_{j=1}^{N_1} X_{j,1}, \ldots, \sum_{j=1}^{N_m} X_{j,m}  \right). \nonumber
\end{equation}
If $\bm{X}_j  \le_{dcx} \bm{Y}_j $, $\forall j$, then
\begin{equation}
\left( \sum_{j=1}^{N_1} X_{j,1}, \ldots, \sum_{j=1}^{N_m} X_{j,m}  \right) 
 \le_{dcx} \left( \sum_{j=1}^{N_1} Y_{j,1}, \ldots, \sum_{j=1}^{N_m} Y_{j,m}  \right). \nonumber
\end{equation}
If $\bm{M} \le_{dcx} \bm{N}$ and $\bm{X}_j  \le_{dcx} \bm{Y}_j $, $\forall j$, then
\begin{equation}
\left( \sum_{j=1}^{M_1} X_{j,1}, \ldots, \sum_{j=1}^{M_m} X_{j,m}  \right) \le_{dcx}
\left( \sum_{j=1}^{N_1} Y_{j,1}, \ldots, \sum_{j=1}^{N_m} Y_{j,m}  \right). \nonumber
\end{equation}
\end{theorem}

\begin{proof}
The first result is available in \cite{pellerey1999stochastic}\cite{shaked2007stochastic}.
For the second result, $\bm{X}_j  \le_{dcx} \bm{Y}_j \implies X_{j,i} \le_{cx} Y_{j,i}$, $\forall 1\le i\le m$, the independence assumption implies that the convex order is closed under convolutions \cite{shaked2007stochastic}, i.e., $\sum_{j=1}^{n_i} X_{j,i} \le_{cx} \sum_{j=1}^{n_i} Y_{j,i}$, $\forall 1\le i\le m$, furthermore, it implies
$
\mathbb{E} \qty[ \phi \left( \sum_{j=1}^{N_1} X_{j,1}, \ldots, \sum_{j=1}^{N_m} X_{j,m}  \right) | \bm{N} = \qty(n_1, \ldots, n_m) ]
 \le \mathbb{E} \qty[ \phi \left( \sum_{j=1}^{N_1} Y_{j,1}, \ldots, \sum_{j=1}^{N_m} Y_{j,m}  \right) | \bm{N} = \qty(n_1, \ldots, n_m) ],
$
where $\phi$ is directionally convex. By integrating for expectation, we obtain the final result.
The third result follows the transitivity of the directionally convex order.
\end{proof}

\begin{corollary}
With proper revisions, the results in Theorem \ref{theorem-random-multiplexing-dcx} extend to the increasing directionally convex order $\le_{idcx}$ and (increasing) componentwise convex order ($\le_{iccx}$) $\le_{ccx}$.
Specifically, the corresponding revisions are $X_{j,i} \le_{icx} (\le_{cx}, \le_{icx})  X_{j+1,i}$, $\bm{M} \le_{idcx} (\le_{ccx}, \le_{iccx}) \bm{N}$, and $\bm{X}_j  \le_{idcx} (\le_{ccx}, \le_{iccx}) \bm{Y}_j$.
\end{corollary}

\begin{proof}
The proof follows the properties of each stochastic orders and preliminary results in \cite{pellerey1999stochastic}\cite{shaked2007stochastic}.
\end{proof}

\begin{remark}
It is interesting to notice the fact that: Let $\bm{X} = \qty(X_1, \ldots, X_m )$ be a set of independent random variables and let $\bm{Y} = \qty( Y_1, \ldots, Y_m )$ be another set of independent random variables, then, $\bm{X} \le_{dcx} (\le_{idcx}) \bm{Y} \iff X_i \le_{cx} (\le_{icx}) Y_i,\ \forall 1\le i\le m \iff \bm{X} \le_{ccx} (\le_{iccx}) \bm{Y}$.
\end{remark}

We study the ordering property of the partial sums under the ordering condition of the sequences.

\begin{theorem}
For the stochastic process, if $\qty(X_{t_1}, \ldots, X_{t_k}) \le_{idcx} \qty(\widetilde{X}_{t_1}, \ldots, \widetilde{X}_{t_k})$, $\forall t_1, \ldots, t_k \in \mathbb{N}$, $\forall k\in \mathbb{N}$, then, we have
\begin{equation}
\qty( \sum_{j_1 \in \mathcal{T}_1} X_{j_1}, \ldots, \sum_{j_k \in \mathcal{T}_k} X_{j_k} ) 
\le_{idcx}
\qty( \sum_{j_1 \in \mathcal{T}_1} \widetilde{X}_{j_1}, \ldots, \sum_{j_k \in \mathcal{T}_k} \widetilde{X}_{j_k} ),
\end{equation}
for any disjoint subsets $\mathcal{T}_1, \ldots, \mathcal{T}_k \in \mathbb{N}$. 
\end{theorem}

\begin{proof}
The proof follows that, if $f : \mathbb{R}^m \rightarrow \mathbb{R}^k$ is increasing and directionally convex and $g : \mathbb{R}^n \rightarrow \mathbb{R}^m$ is increasing and directionally convex, then the composition $f\circ g$ is increasing and directionally convex \cite{shaked2007stochastic}.
\end{proof}

\begin{remark}
An alternative approach is to treat the functional stochastic process as a random field on $\mathbb{N}^n \times \mathbb{R}$, then the comparison result directly follows the comparison result of random field in \cite{miyoshi2004note}\cite{shaked2007stochastic}.
\end{remark}

\begin{remark}
The result indicates that the $\le_{idcx}$ ordering of the instantaneous values implies the $\le_{idcx}$ ordering of the accumulated values.
\end{remark}

\begin{remark}
It is interesting to study the corresponding property of the supermodular order or the counter examples.
\end{remark}

Since the usual stochastic order has a direct indication on the mean values, i.e., $\bm{X} \le_{st} \bm{Y} \implies \mathbb{E}\bm{X} \le \mathbb{E}\bm{Y} \implies \sum\mathbb{E}{X}_i \le \sum\mathbb{E}{Y}_i$, $X_i \in \bm{X},\ Y_i \in \bm{Y}$, it is interesting to consider the dependence manipulation with respect to the usual stochastic order when the mean value is the objective measure.

\begin{theorem}\label{theorem-marginal-st}
Assume the random parameters are spatially independent and temporally dependent.
If $f_j\qty( X_j^1\ldots, X_j^n )$, $\forall 1\le j\le t$, are increasing, 
and
$
\left( X_1^i, X_2^i, \ldots, X_t^i \right) \le_{st} \left( \widetilde{X}_1^i, \widetilde{X}_2^i, \ldots, \widetilde{X}_t^i \right), 
$ 
$\forall 1\le i\le n$, 
then
\begin{equation}
 \left( X_1, X_2, \ldots, X_t \right) \le_{st} \left( {\widetilde{X}}_1, {\widetilde{X}}_2, \ldots, {\widetilde{X}}_t \right), 
\end{equation}
where ${\widetilde{X}}_j = f_j ( X_j^1, \ldots, X_j^{i-1}, \widetilde{X}_j^i, X_j^{i+1}, \ldots, X_j^n )$, $\forall 1\le j\le t$, for any $1\le i \le n$;
and
\begin{equation}
\widetilde{\bm{X}}_{t}^{k} \le_{st} \widetilde{\bm{X}}_{t}^{k^\prime},\ \forall 0\le k \le k^\prime \le n, 
\end{equation}
where $\widetilde{{X}}_{{t}_m}^{l} = f_m ( \widetilde{X}_m^1, \ldots, \widetilde{X}_m^l, {X}_m^{l+1}, \ldots, {X}_m^n )$, $\ 1\le m \le t$, and $\widetilde{\bm{X}}_{t}^{l} = ( \widetilde{{X}}_{t_1}^{l}, \ldots, \widetilde{{X}}_{t_t}^{l} )$, $l\in\{k,k^\prime \}$.
\end{theorem}

\begin{proof}
The spatial independence implies the conjunction $\qty( \bm{X}^1, \ldots, \bm{X}^i, \ldots, \bm{X}^n ) \le_{st} \qty( \bm{X}^1, \ldots, \widetilde{\bm{X}}^i, \ldots, \bm{X}^n )$, where $\bm{X}^i = \left( X_1^i, X_2^i, \ldots, X_t^i \right)$, then the first result directly follows the closure property of the usual stochastic order \cite{shaked2007stochastic}.
The second result follows the transitivity of the usual stochastic order.
\end{proof}

\begin{remark}
Particularly, if $\left( X_1^i, X_2^i, \ldots, X_t^i \right)$ and $\left( {\widetilde{X}}_1^i, {\widetilde{X}}_2^i, \ldots, {\widetilde{X}}_t^i \right)$ have the common copula $C^i_t \qty(u_1^i,\ldots, u_t^i)$, the order condition $\left( X_1^i, X_2^i, \ldots, X_t^i \right) \le_{st} \left( \widetilde{X}_1^i, \widetilde{X}_2^i, \ldots, \widetilde{X}_t^i \right)$ can be replaced by $X_j^i \le_{st} \widetilde{X}_j^i$, $\forall 1\le j \le t$.
This result is available in \cite[p. 272]{shaked2007stochastic}.
\end{remark}

\begin{corollary}
Assume the random parameters are spatially dependent and temporally independent.
If $f_j\qty( X_j^1\ldots, X_j^n )$, $\forall 1\le j\le t$, are increasing, 
and
$
\left( X_j^1, X_j^2, \ldots, X_j^n \right) \le_{st} \left( \widetilde{X}_j^1, \widetilde{X}_j^2, \ldots, \widetilde{X}_j^n \right), 
$ 
$\forall 1\le j\le t$, 
then the same results hold as in Theorem \ref{theorem-sm-sp-dependent}, but in the sense of the usual stochastic order $\le_{st}$.
\end{corollary}

\begin{proof}
The results follow the closure property and transitivity of the usual stochastic order \cite{shaked2007stochastic}.
\end{proof}

\begin{remark}
It is interesting to extend the results to the scenario without spatial and temporal dependence specification and express the results in terms of the conditional probability and conditional stochastic order as in Theorem \ref{theorem-sm-dependence}.
\end{remark}

We have the following results of the random sums with respect to the usual stochastic order.

\begin{remark}
Let $\bm{X}_j = ( X_{j,1}, \ldots, X_{j,m} )$ and $\bm{Y}_j = ( Y_{j,1}, \ldots, Y_{j,m} )$, $j = 1,2,\ldots$, be two sequences of non-negative random vectors, and let $\bm{M} = \left(M_1, M_2, \ldots , M_m \right)$ and $\bm{N} = \left( N_1, N_2, \ldots, N_m \right)$ be two vectors of non-negative integer-valued random variables. Assume that both $\bm{M}$ and $\bm{N}$ are independent of the $\bm{X}_j$'s and $\bm{Y}_j$'s. 

If $\qty{ \bm{X}_j, j \in \mathbb{N} } \le_{st} \qty{ \bm{Y}_j, j \in \mathbb{N} }$ and $\bm{M} \le_{st} \bm{N}$, then
\begin{equation}
\left( \sum_{j=1}^{M_1} X_{j,1}, \ldots, \sum_{j=1}^{M_m} X_{j,m}  \right) 
\le_{st} \left( \sum_{j=1}^{N_1} Y_{j,1}, \ldots, \sum_{j=1}^{N_m} Y_{j,m}  \right). \nonumber
\end{equation}
This result is available in \cite{shaked2007stochastic}. Specifically, the proof follows the transitivity of the following results,
\begin{equation}
\left( \sum_{j=1}^{M_1} X_{j,1}, \ldots, \sum_{j=1}^{M_m} X_{j,m}  \right) 
\le_{st} \left( \sum_{j=1}^{N_1} X_{j,1}, \ldots, \sum_{j=1}^{N_m} X_{j,m}  \right), \nonumber
\end{equation}
which is provided in \cite{pellerey1999stochastic}, and
\begin{equation}
\left( \sum_{j=1}^{N_1} X_{j,1}, \ldots, \sum_{j=1}^{N_m} X_{j,m}  \right) 
 \le_{st} \left( \sum_{j=1}^{N_1} Y_{j,1}, \ldots, \sum_{j=1}^{N_m} Y_{j,m}  \right), \nonumber
\end{equation}
which follows the closure property of the usual stochastic order, by conditioning on $\qty(N_1, \ldots, N_m) = \qty(n_1, \ldots, n_m)$ and integrating for expectation. 
\end{remark}

\section{Conclusion}
\label{conclusion}

This paper provides an extreme value perspective on the wireless channel gain and the wireless channel capacity.
Specifically, the theoretical results reason why the typical statistical distributions are useful for wireless channel modeling in terms of the tail behavior, i.e.,
why the Rayleigh, Rice, and Nakagami distributions can model the small-scale fading and why the Lognormal distribution can model the large-scale fading.
We highlight that the passive nature of the wireless propagation environment results in that the wireless channel gain has finite moments in the stochastic channel models.
Considering the power constraints in the wireless communication system, we show that the light-tailed behavior is an intrinsic property of the wireless channel capacity.

As a hypothesis, the capacity is possible to be heavy-tailed in theory.
Considering a generalized channel scenario, the capacity formula may include many random parameters, specifically, if a random parameter has a heavy or super-heavy tail and occurs as sum or product with other random parameters in the capacity formula, e.g., the compound capacity as sum of sub-capacities, then this random parameter may dominate the tail behavior and result in a heavy-tailed capacity distribution.
Thus, it is interesting to investigate the empirical tail behavior of the wireless channel capacity in the real wireless systems or wireless networks, e.g., the multi-user scenario of multiple access channel and broadcast channel.
Moreover, it is interesting to take into account the tail behavior of the noise process.

In addition, this paper studies the dependence transformability with respect to both the marginal distributions and the dependence structures. 
The results with respect to the directionally convex order rely on the strictly increasing affine functions in the proof, it is interesting to extend the results to more general function spaces.
It is interesting to note that the marginal manipulation has a dependence bias, i.e., it has an effective effect with respect to positive dependence rather than negative dependence. 
This property indicates that the dependence is a tradable resource in the physical world and it sheds light on the development of new wireless technologies to trade off the dependence resource and other resources, e.g., exchanging dependence for transmission power.
By analogy with power allocation, it is interesting to define new dependence measures and study the mechanism of dependence allocation.

\appendices
\section{Proof of Lemma \ref{lemma-tail-equivalence}}
\label{proof-lemma-tail-equivalence}

\begin{enumerate}[wide, labelwidth=!, labelindent=0pt]
\item
Considering the structure of the capacity formula, $c =a\log_{2}(1+b\lambda_{\max})$, we have
$
\overline{F}_{c}(x)=O \qty(e^{-\theta x}) 
\iff \overline{F}_{\lambda_{\max}}(x)=O \qty(x^{-\theta}) 
\iff \overline{F}_{\Tr\qty[\bm{H}\bm{H}^\ast]}(x)=O \qty(x^{-\theta}).
$
The first relationship follows the transform of random variables, i.e., $\overline{F}_{c}(x)=O \qty(e^{-\theta c}) \iff \overline{F}_{\lambda_{\max}}(x)=O \qty( (1 + b x)^{- \theta} ) $, $(1 + b x)^{- \theta} \sim  (bx)^{- \theta} $, and $\overline{F}_{\lambda_{\max}}(x)=O \qty(x^{-\theta}) \iff \overline{F}_{\lambda_{\max}}(x)=O \qty((bx)^{-\theta}) $. 
The second relationship follows that $\lambda_{\max}\le \Tr\qty[\bm{H}\bm{H}^\ast] \implies \overline{F}_{\lambda_{\max}}(x) \le \overline{F}_{\Tr\qty[\bm{H}\bm{H}^\ast]}(x)$, 
thus, $\overline{F}_{\Tr\qty[\bm{H}\bm{H}^\ast]}(x)=O \qty(x^{-\theta}) \implies \overline{F}_{\lambda_{\max}}(x)=O \qty(x^{-\theta})$;
and $r(\bm{H})\lambda_{\max} \ge \Tr\qty[\bm{H}\bm{H}^\ast] \implies \overline{F}_{r(\bm{H})\lambda_{\max}}(x) \ge \overline{F}_{\Tr\qty[\bm{H}\bm{H}^\ast]}(x)$, thus, $\overline{F}_{\lambda_{\max}}(x)=O \qty(x^{-\theta}) \implies \overline{F}_{\Tr\qty[\bm{H}\bm{H}^\ast]}(x)=O \qty(x^{-\theta})$.

\item
First, we have the inequality, $\mathbb{E} \qty[ \qty( 1 + {\Delta}  \lambda_{\max}  )^\theta ] \le (1+\Delta)^\theta \mathbb{E} \qty[ \qty( 1 + \lambda_{\max}  )^\theta ]$, which implies that 
$
\mathbb{E} \qty[ \qty( 1 + \lambda_{\max}  )^\theta ] < \infty, \ \exists \theta>0
\implies
\mathbb{E} \qty[ \qty( 1 + {\Delta}  \lambda_{\max}  )^\theta ] < \infty, \ \exists \theta>0
$, 
because $(1+\Delta)^\theta<\infty$.
Second, for $0<\Delta<1$, letting $\Delta\ge \frac{1}{m}$, $m\in\mathbb{N}$,
we have 
$
m^\theta \mathbb{E} \qty[ \qty( 1 + \Delta \lambda_{\max}  )^\theta ] \ge
\mathbb{E} \qty[ \qty( 1 + \lambda_{\max}  )^\theta ]
$,
which implies that
$
\mathbb{E} \qty[ \qty( 1 + \Delta \lambda_{\max}  )^\theta ] < \infty, \ \exists \theta>0
\implies
\mathbb{E} \qty[ \qty( 1 + \lambda_{\max}  )^\theta ] < \infty, \ \exists \theta>0
$,
because $m^\theta <\infty$;
for $\Delta\ge 1$, it is trivial.

Since $\qty(\lambda_{\max})^\theta \le \qty(1 + \lambda_{\max})^\theta$, we have 
$\mathbb{E} \qty[ \qty( 1 + \lambda_{\max}  )^\theta ] < \infty, \ \exists \theta>0
\implies
\mathbb{E} \qty[ \qty( \lambda_{\max}  )^\theta ] < \infty, \ \exists \theta>0$.
According to H\"{o}lder's inequality, $\mathbb{E}\qty[X^r] \le \qty(\mathbb{E}\qty[X^s])^{r/s}$, $ 0< r< s$, $X\in \mathbb{R}_{\ge 0}$, which implies that $\mathbb{E}\qty[X^s]< \infty \implies \mathbb{E}\qty[X^r] <\infty$.
Specifically, we have 
$\mathbb{E} \qty[ \qty(  \lambda_{\max}  )^\theta ] < \infty, \ \exists \theta>1
\implies
\mathbb{E} \qty[ \qty( \lambda_{\max}  )^\theta ] < \infty, \ \forall 0< \theta \le 1
$.
Suppose $\mathbb{E} \qty[ \qty( \lambda_{\max}  )^\theta ] < \infty, \ \exists 0< \theta \le 1$, we have $\overline{F}_{\lambda_{\max}} (x) = o\qty(x^{-\theta})$ and $\int_{0}^{\infty} \overline{F}_{\lambda_{\max}}(x) x^{\theta-1} d x <\infty $,
since $x^{-\theta} \sim (1+x)^{-\theta}$ and $x^{\theta-1} \ge (1+x)^{\theta-1} $, we further have $\overline{F}_{\lambda_{\max}} (x) = o\qty((1+x)^{-\theta})$ and $\int_{0}^{\infty} \overline{F}_{\lambda_{\max}}(x) (1+x)^{\theta-1} d x <\infty $, which corresponds to $\mathbb{E} \qty[ \qty( 1 + \lambda_{\max}  )^\theta ] < \infty$.
Thus, $\mathbb{E} \qty[ \qty( 1 + \lambda_{\max}  )^\theta ] < \infty, \ \exists 0< \theta\le 1
\impliedby
\mathbb{E} \qty[ \qty( \lambda_{\max}  )^\theta ] < \infty, \ \exists 0< \theta\le 1$.
In all, we obtain 
$\mathbb{E} \qty[ \qty( 1 + \lambda_{\max}  )^\theta ] < \infty, \ \exists \theta>0
\iff
\mathbb{E} \qty[ \qty( \lambda_{\max}  )^\theta ] < \infty, \ \exists \theta>0$.

If $\mathbb{E}\qty[\lambda_{\max}] = \infty$, according to Jensen's inequality, $\mathbb{E} \qty[ \qty( \lambda_{\max}  )^\theta ] \ge \qty(\mathbb{E}\qty[\lambda_{\max}])^\theta = \infty$, $\forall\theta \ge 1$,
which implies that $\mathbb{E}\qty[ \qty( 1 + \lambda_{\max})^\theta ] = \infty$, $\forall \theta \ge 1$, because $\mathbb{E}\qty[ \qty( \lambda_{\max})^\theta ] \le \mathbb{E}\qty[ \qty( 1 + \lambda_{\max})^\theta ]$.

\item
Considering that the matrix $\bm{H}\bm{H}^\ast \in \mathfrak{X}_{\ge 0}$, we have $\lambda_{\max}\le {\Tr\qty[\bm{H}\bm{H}^\ast]}$ and ${\Tr\qty[\bm{H}\bm{H}^\ast]}\le r(\bm{H}) \lambda_{\max}$. Then, the proof follows.
\end{enumerate}

\section{Proof of Theorem \ref{theorem-tail-sufficient}}
\label{proof-theorem-tail-sufficient}

First, consider the scenario with channel side information only at the receiver.
We denote $\rho := \frac{P}{N_0 W}$.
On the one hand, we have an upper bound of the capacity
\begin{IEEEeqnarray}{rCl}
c &=& W \sum\nolimits_{i=1}^{r(\bm{H} )} \log_2 \qty( 1 + \frac{\rho}{N_T} \lambda_i ) \\
&\le& W r(\bm{H} ) \log_2 \qty( 1 + \frac{\rho}{N_T}  \lambda_{\max}  ),
\end{IEEEeqnarray}
where $r(\bm{H})$ is the rank of matrix $\bm{H}$, $\lambda_i$ is the eigenvalue of the matrix $\bm{H}\bm{H}^\ast$, and the equality follows the eigenvalue expression of the capacity.

Second, consider the scenario with full channel side information.
The capacity is upper bounded by
\begin{IEEEeqnarray}{rCl}
c &=& W \max_{\sum_{i=1}^{r(\bm{H})} \gamma_i = N_T} \sum\nolimits_{i=1}^{r(\bm{H})} \log_2 \qty( 1 + \frac{\rho \gamma_i}{N_T} \lambda_i ) \\
&\le& W r(\bm{H}) \log_2 \qty( 1 + \frac{\rho}{N_T} \gamma_{\max}  \lambda_{i}  ) \\
&\le& W r(\bm{H}) \log_2 \qty( 1 + {\rho}  \lambda_{\max}  ).
\end{IEEEeqnarray}

It is easy to show that $\mathbb{E} \qty[ e^{\theta c} ]<\infty,\ \exists \theta>0$, entails $ \mathbb{E} \qty[ \qty( 1 + {\rho}  \lambda_{\max}  )^\theta ] < \infty, \ \exists \theta>0$.

\section{Proof of Theorem \ref{theorem-frequency-selective-sufficient}}
\label{proof-theorem-frequency-selective-sufficient}

If the channel side information is unknown to the transmitter, then \cite{paulraj2003introduction}
\begin{equation}
c = \frac{W}{N} \sum_{i=1}^{N} \log_2 \det \qty( \bm{I}_{N_R} + \frac{\rho}{N_T} \bm{{H}}_i \bm{{H}}_i^\ast ).
\end{equation}
Since the light-tailed property is preserved under the sum operation, if the capacity distribution of each sub-channel is light-tailed, so is the total capacity distribution.

If the channel side information is known to the transmitter, then \cite{paulraj2003introduction}
\begin{IEEEeqnarray}{rCl}
c &=& \frac{W}{N} \max_{\sum_{i=1}^{r(\bm{\mathcal{H}})} \gamma_i = N_T N} \sum_{i=1}^{r(\bm{\mathcal{H}})}  \log_2  \qty( 1 + \frac{\rho \gamma_i}{N_T} \lambda_i\qty( \bm{\mathcal{H}}  \bm{\mathcal{H}}^\ast ) )   \IEEEeqnarraynumspace\\
&\le& \frac{W}{N} \sum_{i=1}^{r(\bm{\mathcal{H}})}  \log_2  \qty( 1 + \frac{\rho \gamma_{\max}}{N_T} \lambda_{i}\qty( \bm{\mathcal{H}}  \bm{\mathcal{H}}^\ast ) ) \\
&\le& \frac{W}{N} {r(\bm{\mathcal{H}})}  \log_2  \qty( 1 + {\rho N} \lambda_{\max} \qty( \bm{\mathcal{H}}  \bm{\mathcal{H}}^\ast ) ).
\end{IEEEeqnarray}
Particularly, we have $\lambda_{\max}= \max\qty(\lambda^{1}_{\max}, \ldots, \lambda^{N}_{\max})$, which implies the equivalent expression of the condition.

\section{Proof of Theorem \ref{theorem-s-n-csi-transmitter}}
\label{proof-theorem-s-n-csi-transmitter}

For the flat channel without channel side information at the transmitter,
the capacity is expressed as \cite{paulraj2003introduction}
\begin{IEEEeqnarray}{rCl}
c &=& W \log_{2} \det \qty( \bm{I}_{N_R} + \frac{\rho}{N_T} \bm{H}\bm{H}^\ast ) \\
&=& W \log_{2} \det \qty( \bm{I}_{N_R} + \frac{\rho}{N_T} \bm{\Lambda} ) \\ 
&=& W \sum\nolimits_{i=1}^{r(\bm{H} )} \log_2 \qty( 1 + \frac{\rho}{N_T} \lambda_i ),
\end{IEEEeqnarray}
where $\bm{H}\bm{H}^\ast = \bm{Q}\bm{\Lambda}\bm{Q}^\ast$ and $\bm{Q}\bm{Q}^\ast = \bm{I}$, the second equality follows $\det(\bm{I}_m + \bm{A}\bm{B}) = \det(\bm{I}_n + \bm{B}\bm{A})$ for $\bm{A}\in \mathbb{C}^{m\times n}$ and $\bm{B}\in \mathbb{C}^{n\times m}$, and the third equality is an equivalent expression.

The proof follows the light-tailed distribution definition, i.e., $\mathbb{E} \qty[ e^{ \theta W \log_{2} \det \qty( \bm{I}_{N_R} + \frac{\rho}{N_T} \bm{\Lambda} ) } ] < \infty$, $\exists \theta >0$.
Specifically, considering finite rank matrix, we have, for $\Delta \ge 1$, $\Delta^{r(\bm{H})\theta} \mathbb{E} \qty[  \prod_{i=1}^{r(\bm{H})} \qty( 1+ \lambda_i )^\theta  ]\ge \mathbb{E} \qty[  \prod_{i=1}^{r(\bm{H})} \qty( 1+ \Delta \lambda_i )^\theta ] \ge \mathbb{E} \qty[  \prod_{i=1}^{r(\bm{H})} \qty( 1+ \lambda_i )^\theta  ]$; 
and for $0< \Delta \le 1$, $\Delta^{r(\bm{H})\theta} \mathbb{E} \qty[  \prod_{i=1}^{r(\bm{H})} \qty( 1+ \lambda_i )^\theta  ] \le \mathbb{E} \qty[  \prod_{i=1}^{r(\bm{H})} \qty( 1+ \Delta \lambda_i )^\theta ] \le \mathbb{E} \qty[  \prod_{i=1}^{r(\bm{H})} \qty( 1+ \lambda_i )^\theta  ]$.
Thus, we have the following equivalent expressions
\begin{multline}
\mathbb{E} \qty[  \prod_{i=1}^{r(\bm{H})} \qty( 1+ \Delta \lambda_i )^\theta  ] < \infty,\ 0<\Delta<\infty, \ \exists \theta >0 \\
\iff
\mathbb{E} \qty[  \prod_{i=1}^{r(\bm{H})} \qty( 1+ \lambda_i )^\theta  ] < \infty,\ \exists \theta >0. 
\end{multline}
This completes the proof.

\section{Proof of Theorem \ref{theorem-sufficient-map}}
\label{proof-theorem-sufficient-map}

We present the proof of the deterministic power scenario, and the extension to the random power scenario is to replace the matrix $\bm{H}\bm{H}^\ast$ with the scalar multiplication $ {p} \bm{H}\bm{H}^\ast$ of the random variable $p$ and the random matrix $\bm{H}\bm{H}^\ast$.

We have a sufficient condition for the light-tailed property, i.e.,
$
\mathbb{E} \qty[ \qty( 1 + \lambda_{\max}  )^\theta ] < \infty, \ \exists \theta>0,
$
which can be relaxed to be
$
\mathbb{E} \qty[ \qty( 1 +  \Tr\qty[ \bm{H}\bm{H}^\ast ]  )^\theta ] < \infty, \ \exists \theta>0,
$
which is equivalent to $\mathbb{E} \qty[ \qty( 1 + \sum_{i=1}^{r(\bm{H})}\lambda_{i}  )^\theta ] < \infty, \ \exists \theta>0$.

Furthermore, the condition can be relaxed to be
$
\mathbb{E} \qty[ \qty(  \Tr\qty[ \bm{I} + \bm{H}\bm{H}^\ast ]  )^\theta ] < \infty, \ \exists \theta>0.
$

Considering the transfer rule \cite{tropp2015introduction}, the function inequality $(1+x)^\vartheta \le e^{\vartheta x}$, $\vartheta>0$, implies the matrix ordering $\qty( \bm{I} + \bm{H}\bm{H}^\ast )^\vartheta \preceq e^{\vartheta  \bm{H}\bm{H}^\ast }$, thus $\Tr\qty[ e^{\vartheta  \bm{H}\bm{H}^\ast } - \qty( \bm{I} + \bm{H}\bm{H}^\ast )^\vartheta ] \ge 0$, equivalently,
$
\Tr\qty[ \qty( \bm{I} + \bm{H}\bm{H}^\ast )^\vartheta ] \le \Tr\qty[ e^{\vartheta  \bm{H}\bm{H}^\ast } ],
$
which implies that the condition is further relaxed to be
$
\mathbb{E} \qty[ \qty(  \Tr\qty[ e^{ \bm{H}\bm{H}^\ast } ]  )^\theta ] < \infty, \ \exists \theta>0.
$

In addition, according to the function inequality, $\qty( 1 + \rho\lambda_{\max}  )^\vartheta \le e^{\vartheta  \rho\lambda_{\max}} \le \Tr\qty[ e^{\vartheta \rho  \bm{H}\bm{H}^\ast } ]$, where $\vartheta>0$ and the last inequality follows that the spectral mapping theorem, thus we have the relaxed condition
$
\mathbb{E} \qty [ \Tr \qty[ e^{\theta \bm{H} \bm{H}^\ast  } ] ] < \infty,\ \exists \theta>0.
$
Similarly, $\qty( 1 + \Tr\qty[ \bm{H}\bm{H}^\ast ]  )^\vartheta \le e^{\vartheta  \Tr\qty[ \bm{H}\bm{H}^\ast ]}$, $\vartheta>0$, thus we have the relaxed condition
$
\mathbb{E} \qty[ e^{\theta  \Tr\qty[ \bm{H}\bm{H}^\ast ]} ] < \infty, \ \exists \theta >0,
$
which is equivalent to $\mathbb{E} \qty[ e^{\theta  \lambda_{\max} } ] < \infty, \ \exists \theta >0$.

\section{Proof of Lemma \ref{lemma-tail-relationships}}
\label{proof-lemma-tail-relationships}

The proofs follow the given conditions and the definition of the asymptotic symbols.

\begin{enumerate}[wide, labelwidth=!, labelindent=0pt]
\item
By the given condition, we have $\lim\limits_{x\rightarrow \infty} \frac{ \overline{F}_{X_1}(x - \varphi(x)) }{ \overline{F}_{X_1}(x) } =1 $, 
$\forall \varphi \in \mathfrak{F}$.
Thus, $\overline{F}_{X_1}\qty(x-\varphi(x)) \sim \overline{F}_{X_1}(x) $.
Considering $F\circ \log \in \mathcal{R}_{0}$, $\forall F \in \mathcal{L}$, we have $\lim\limits_{x\rightarrow \infty} \frac{\overline{F}_{X_1}\qty( \log \frac{x}{\varphi(x)} ) }{ \overline{F}_{X_1} \qty( \log x ) } = \lim\limits_{x\rightarrow \infty} \qty( \frac{ \log \frac{x}{\varphi(x)} }{ \log x } )^{-\theta_1} = \lim\limits_{x\rightarrow \infty} \qty( 1- \frac{\log \varphi(x)}{ \log x } )^{-\theta_1} = 1$, $\exists \theta_1 > 0$, $\forall \varphi \in \mathfrak{F}$.
Thus, $\overline{F}_{X_1}\qty( \log \frac{ x }{\varphi(x)}) \sim  \overline{F}_{X_1}(\log x) $.

Since   $\overline{F}_{X_1}(x) = \Omega\qty(x^{-\theta_1 }) $, $ \theta_1>0$ $\iff$ $\exists x_1 >0$, $\forall x>x_1$, $\exists C_1 >0$, $\overline{F}_{X_1}(x) \ge C_1 x^{-\theta_1 } $; $\overline{F}_{X_2}(x) = O\qty(e^{-\theta_2 x}) $, $ \theta_2>0$ $\iff$ $\exists x_2 >0$, $\forall x>x_2$, $\exists C_2 >0$, $\overline{F}_{X_2}(x) \le C_2 e^{-\theta_2 x} $. Let $x_0 = \max(x_1, x_2)$, then,
\begin{multline}
\limsup\limits_{x\rightarrow \infty} \frac{ \overline{F}_{X_2}\qty(\varphi(x)) }{ \overline{F}_{X_1}(x) }
\le \inf\limits_{x^\ast > x_0}\sup\limits_{x \ge x^\ast} \frac{ C_2 e^{-\theta_2 \varphi(x)} }{ C_1 x^{-\theta_1 } } \\
=  \inf\limits_{x^\ast > x_0}\sup\limits_{x \ge x^\ast} \frac{ C_2 }{ C_1 } \frac{ x^{\theta_1} }{ \sum_{n=0}^{\infty} \frac{ \qty( \theta_2 x^\alpha )^n }{n!}  } =0.  
\end{multline}
Thus, we obtain $\overline{F}_{X_2}\qty(\varphi(x)) = O\qty( \overline{F}_{X_1}(x) )$.

Similarly, if $\overline{F}_{X_1}(x) = \Omega\qty( x^{-\theta_1} )$, $\exists \theta_1 >0$, and $\overline{F}_{X_2}(x) = O\qty(x^{-\theta_2 }) $, $\exists \theta_2>0$, $\varphi(x)=x^\alpha$, $0<\alpha <1$, and $\alpha \theta_2 > \theta_1$, then
\begin{multline}
\limsup\limits_{x\rightarrow \infty} \frac{ \overline{F}_{X_2}\qty(\varphi(x)) }{ \overline{F}_{X_1}(x) }
\le \inf\limits_{x^\ast > x_0}\sup\limits_{x \ge x^\ast} \frac{ C_2 \qty(\varphi(x))^{-\theta_2 } }{ C_1 x^{-\theta_1 } } \\
=  \inf\limits_{x^\ast > x_0}\sup\limits_{x \ge x^\ast} \frac{ C_2 }{ C_1 } \frac{ x^{\theta_1} }{ x^{\alpha \theta_2}  } =0.  
\end{multline}
When $\alpha \theta_2 = \theta_1$, the limit is also finite. Thus, we obtain $\overline{F}_{X_2}\qty(\varphi(x)) = O\qty( \overline{F}_{X_1}(x) )$.

\item
By the given condition, we have $\lim\limits_{x\rightarrow \infty} \frac{ \overline{F}_{X_1}\qty(x-\varphi(x)) }{ \overline{F}_{X_1}(x) } = \lim\limits_{x\rightarrow \infty} \frac{ L_1 \qty(x - \varphi(x) ) }{ L_1(x)}  \qty( 1-\frac{\varphi(x)}{x} )^{\theta_1} =1$, $\forall \varphi \in \mathfrak{F}$, because we obtain $t=1$ for $x-\varphi(x) = t x$ as $x\rightarrow \infty$.
Thus, we obtain $\overline{F}_{X_1}\qty(x-\varphi(x)) \sim \overline{F}_{X_1}(x) $.
Similarly, we have $\lim\limits_{x\rightarrow \infty} \frac{ \overline{F}_{X_1}\qty(\frac{x}{ \varphi(x) } ) }{ \overline{F}_{X_1}(x) } = \lim\limits_{x\rightarrow \infty} \frac{L_1 \qty( \frac{x}{\varphi(x)} )}{ L_1(x) } \qty(\varphi(x))^{\theta_1 }$, $\forall \varphi \in \mathfrak{F}$.
Thus, $\overline{F}_{X_1}\qty( \frac{ x}{\varphi(x)} )= \omega\qty( \overline{F}_{X_1}(x) ) $, $\forall \theta_1 > 0$ and $\overline{F}_{X_1}\qty( \frac{ x}{\varphi(x)} ) \sim \overline{F}_{X_1}(x) $ for $\theta_1 = 0$.

Since $\overline{F}_{X_2}(x) = O\qty(e^{-\theta_2 x}) $, $\theta_2>0$ $\iff$ $\exists x_0 >0$, $\forall x>x_0$, $\exists C_2 >0$, $\overline{F}_{X_2}(x) \le C_2 e^{-\theta_2 x} $. Then,
\begin{multline}
\limsup\limits_{x\rightarrow \infty} \frac{ \overline{F}_{X_2}\qty(\varphi(x)) }{ \overline{F}_{X_1}(x) }
\le \inf\limits_{x^\ast > x_0}\sup\limits_{x \ge x^\ast} \frac{ C_2 e^{-\theta_2 \varphi(x)} }{ L_1(x) x^{-\theta_1} } \\
=  \inf\limits_{x^\ast > x_0}\sup\limits_{x \ge x^\ast} \frac{ C_2 }{ L_1(x) } \frac{ x^{\theta_1} }{ \sum_{n=0}^{\infty} \frac{ \qty( \theta_2 x^\alpha )^n }{n!}  } =0,
\end{multline}
where $\lim_{x\rightarrow\infty} x^{\epsilon}L_1(x) = \infty$, $\forall \epsilon>0$, follows the representation theorem of the slowly varying function.
Thus, we obtain $\overline{F}_{X_2}\qty(\varphi(x)) = O\qty( \overline{F}_{X_1}(x) )$.

Similarly, for $\overline{F}_{X_2}(x) = O\qty(x^{-\theta_2 }) $, $\alpha \theta_2>\theta_1$, we have
$
\limsup\limits_{x\rightarrow \infty} \frac{ \overline{F}_{X_2}\qty(\varphi(x)) }{ \overline{F}_{X_1}(x) }
\le \inf\limits_{x^\ast > x_0}\sup\limits_{x \ge x^\ast} \frac{ C_2 x^{-\theta_2 \alpha} }{ L_1(x) x^{-\theta_1} } = 0.
$

The proof of the rest results follows the previous proofs, by considering the complementary $\liminf\limits_{x\rightarrow \infty}(\cdot)$ and by noticing the fact that, the limit $\lim\limits_{x\rightarrow\infty}(\cdot)$ exists if and only if $\liminf\limits_{x\rightarrow \infty}(\cdot) =\limsup\limits_{x\rightarrow \infty}(\cdot)$.

\item
Since $\overline{F}_{X_1}(x) = \Theta\qty( x^{-\theta_1} )$, $\theta_1 >0$, we have, $\exists C_1^u >0$, $\exists C_1^l >0$, $\exists x_0>0$, $\forall x> x_0$, $C_1^l x^{-\theta_1} \le \overline{F}_{X_1}(x) \le C_1^u x^{-\theta_1} $,
and
$\exists C_1^{u^\prime} >0$, $\exists C_1^{l^\prime} >0$, $\exists x_0^\prime>0$, $\forall x> x_0^\prime$, $C_1^{l^\prime} x^{-\theta_1} \le \overline{F}_{X_1}(x) \le C_1^{u^\prime} x^{-\theta_1} $.
Let $x_0^\ast = \max\qty( x_0, x_0^\prime )$.
Then,
$
\limsup\limits_{x\rightarrow \infty} \frac{ \overline{F}_{X_1}\qty({x} - { \varphi(x) } ) }{ \overline{F}_{X_1}(x) }
\le \inf\limits_{x^\ast > x_0^\ast}\sup\limits_{x \ge x^\ast} \frac{C_1^{u^\prime}}{C_1^{l}} \qty( 1- \frac{\varphi(x)}{x} )^{\theta_1} = \frac{C_1^{u^\prime}}{C_1^{l}},
$ and 
$
\liminf\limits_{x\rightarrow \infty} \frac{ \overline{F}_{X_1}\qty({x} - { \varphi(x) } ) }{ \overline{F}_{X_1}(x) } \ge \frac{C_1^{l^\prime}}{C_1^{u}},
$
where $ \frac{C_1^{l^\prime}}{C_1^{u}} \le \frac{C_1^{u^\prime}}{C_1^{l}}$, because $C_1^l \le C_1^u$ and $C_1^{l^\prime} \le C_1^{u^\prime}$.
Thus, we obtain $ \overline{F}_{X_1}\qty({x} - { \varphi(x) } )  = \Theta \qty( \overline{F}_{X_1}(x) )$.

Since $\overline{F}_{X_1}(x) = \Theta\qty( x^{-\theta_1} )$, $\theta_1 >0$, 
we have 
$
\limsup\limits_{x\rightarrow \infty} \frac{ \overline{F}_{X_1}\qty(\frac{x}{ \varphi(x) } ) }{ \overline{F}_{X_1}(x) }
\le \inf\limits_{x^\ast > x_0}\sup\limits_{x \ge x^\ast} C^\ast \qty(\varphi(x))^{\theta_1} = \infty,
$ $\forall C^\ast >0$, $\forall \varphi \in \mathfrak{F}$.
Similarly, we have 
$
\liminf\limits_{x\rightarrow \infty} \frac{ \overline{F}_{X_1}\qty(\frac{x}{ \varphi(x) } ) }{ \overline{F}_{X_1}(x) } \ge \infty.
$
Thus, we obtain $\overline{F}_{X_1}\qty(\frac{x}{\varphi(x)}) = \omega\qty( \overline{F}_{X_1}(x) )$.

The proofs of the rest results are analogical to the previous proofs, by considering $\lim\limits_{x\rightarrow\infty}(\cdot)$, $\liminf\limits_{x\rightarrow \infty}(\cdot)$, and  $\limsup\limits_{x\rightarrow \infty}(\cdot)$.

\item
Since $ \overline{F}_{X_1}(x)  = \Theta\qty(e^{-\theta_1 x}) $, $\theta_1 >0$, we have $\limsup\limits_{x\rightarrow \infty} \frac{ \overline{F}_{X_1}\qty({x}/{ \varphi(x) } ) }{ \overline{F}_{X_1}(x) } \le \inf\limits_{x^\ast > x_0}\sup\limits_{x \ge x^\ast} C^\ast e^{ -\theta_1 \qty(x/\varphi(x) -x) } = \infty$, $\forall 0< C^\ast < \infty$ and $\forall \varphi \in \mathfrak{F}$; and $\liminf\limits_{x\rightarrow \infty} \frac{ \overline{F}_{X_1}\qty({x}/{ \varphi(x) } ) }{ \overline{F}_{X_1}(x) } \ge \sup\limits_{x^\ast > x_0}\inf\limits_{x \ge x^\ast} C^\star e^{ -\theta_1 \qty(x/\varphi(x) -x) } = \infty$, $\forall 0< C^\star < \infty$ and $\forall \varphi \in \mathfrak{F}$. Thus, $\overline{F}_{X_1}\qty( \frac{ x}{\varphi(x)} )= \omega\qty( \overline{F}_{X_1}(x) ) $.
The proof of the other result follows analogically.

Since $ \overline{F}_{X_i}(x)  = \Theta\qty(e^{-\theta_i x}) $, $\theta_i >0$, $i\in\{1,2\}$, we have
$
\limsup\limits_{x\rightarrow \infty} \frac{ \overline{F}_{X_2} (\varphi(x)) }{ \overline{F}_{X_1}(x) } 
\le \inf\limits_{x^\ast > x_0} \sup\limits_{x\ge x^\ast} C^\ast \frac{e^{ -\theta_2 \varphi(x) }}{ e^{-\theta_1 x} } = \infty,
$
where the last step follows that $\lim\limits_{x\rightarrow \infty}\frac{x}{ \varphi(x) }$.
Similarly, we have 
$
\liminf\limits_{x\rightarrow \infty} \frac{ \overline{F}_{X_2} (\varphi(x)) }{ \overline{F}_{X_1}(x) } \ge \infty.
$
Thus, $\overline{F}_{X_2} (\varphi(x)) = \omega\qty( \overline{F}_{X_1}(x) )$.

Since $ \overline{F}_{X_2}(x)  = \Theta\qty(x^{-\theta_2 }) $, $\theta_2 >0$, we have $\overline{F}_{X_2}(x)  = \omega\qty( \overline{F}_{X_1}(x) ) $, $\forall \theta_2 >0$. Letting $x=\varphi(y)$, then $\overline{F}_{X_2}(\varphi(y))  = \omega\qty(\overline{F}_{X_1}(\varphi(y)) )$, $\forall \varphi \in \mathfrak{F}$. Since $\overline{F}_{X_1}(\varphi(y)) = \overline{F}_{X_1}(y) $,  thus, $\overline{F}_{X_2}(\varphi(y)) = \omega\qty( \overline{F}_{X_1}(y) )$, $\forall \theta_2>0$, $\forall \varphi \in \mathfrak{F}$.
\end{enumerate}

This completes the proofs.

\section{Proof of Theorem \ref{theorem-fat-tail-sufficient}}
\label{proof-theorem-fat-tail-sufficient}

We prove the case of $N=2$ and the proof of the case $N>2$ follows by the iteration of the same procedure.

Considering the independence between $X_1$ and $X_2$, for $x>0$, we have $\overline{F}_{X_1 X_2}(x) = \mathbb{E}\qty[ \overline{F}_{X_1} \qty(\frac{x}{X_2}) ]$, which is reformulated as
\begin{multline}
\mathbb{E}\qty[ \overline{F}_{X_1} \qty(\frac{x}{X_2}) ] = \mathbb{E}\qty[ \overline{F}_{X_1} \qty(\frac{x}{X_2}) 1_{0< X_2 \le \varphi_2(x) } ] \\
+ \mathbb{E}\qty[ \overline{F}_{X_1} \qty(\frac{x}{X_2}) 1_{ X_2 > \varphi_2(x) } ].
\end{multline}

Since $\overline{F}_{X_1}(x) = \Theta\qty(x^{-\theta})$, we have, $\exists x_0>0$, $\exists C_1, C_1^{\prime} >0$, $\forall x> x_0$ and $\forall x/x^{\prime} > x_0$, $\overline{F}_{X_1}(x) \ge C_1 x^{-\theta}$ and $\overline{F}_{X_1}(x/x^{\prime}) \le C_1^{\prime} (x/x^{\prime})^{-\theta}$, and $\limsup\limits_{x\rightarrow \infty} \frac{\overline{F}_{X_1}(x/x^{\prime}) }{\overline{F}_{X_1}(x)} \le \frac{C_1^\prime}{C_1} \qty(x^\prime)^{\theta}$.
Then, 
$
\limsup\limits_{x\rightarrow \infty} \lim\limits_{\varphi_2(x) \rightarrow \infty} 
\int_{(\varphi_2(x),\infty)} \frac{\mathbb{P}\qty( X_1 > x/ x^{\prime} )}{\mathbb{P}\qty( X_1 > x )} \mathbb{P}_{X_2}(d x^\prime)
=\lim\limits_{\varphi_2(x) \rightarrow \infty}
\int_{(\varphi_2(x),\infty)} \limsup\limits_{x\rightarrow \infty} \frac{\mathbb{P}\qty( X_1 > x/ x^{\prime} )}{\mathbb{P}\qty( X_1 > x )} \mathbb{P}_{X_2}(d x^\prime)
\le 0
$.
Similarly, $
\liminf\limits_{x\rightarrow \infty} \lim\limits_{\varphi_2(x) \rightarrow \infty} 
\int_{(\varphi_2(x),\infty)} \frac{\mathbb{P}\qty( X_1 > x/ x^{\prime} )}{\mathbb{P}\qty( X_1 > x )} \mathbb{P}_{X_2}(d x^\prime) \ge 0$.
Thus, $\mathbb{E}\qty[ \overline{F}_{X_1} \qty(\frac{x}{X_2}) 1_{ X_2 > \varphi_2(x) } ] = o\qty(x^{-\theta})$.

Similarly, 
$
\limsup\limits_{x\rightarrow \infty} \lim\limits_{\varphi_2(x) \rightarrow \infty} 
\int_{(0,\varphi_2(x)]} \frac{\mathbb{P}\qty( X_1 > x/ x^{\prime} )}{\mathbb{P}\qty( X_1 > x )} \mathbb{P}_{X_2}(d x^\prime)
=\lim\limits_{\varphi_2(x) \rightarrow \infty}
\int_{(0,\varphi_2(x)]} \limsup\limits_{x\rightarrow \infty} \frac{\mathbb{P}\qty( X_1 > x/ x^{\prime} )}{\mathbb{P}\qty( X_1 > x )} \mathbb{P}_{X_2}(d x^\prime) 
\le C^\ast \mathbb{E} \qty[ \qty(X_2)^{\theta} ]
$, $\exists C^\ast >0$.
On the other hand, $
\liminf\limits_{x\rightarrow \infty} \lim\limits_{\varphi_2(x) \rightarrow \infty} 
\int_{(0,\varphi_2(x)]} \frac{\mathbb{P}\qty( X_1 > x/ x^{\prime} )}{\mathbb{P}\qty( X_1 > x )} \mathbb{P}_{X_2}(d x^\prime)
\ge C^\star \mathbb{E} \qty[ \qty(X_2)^{\theta} ]
$, $\exists C^\star >0$.
Thus, $\mathbb{E}\qty[ \overline{F}_{X_1} \qty(\frac{x}{X_2}) 1_{ 0< X_2 \le \varphi_2(x) } ] = \Theta\qty(x^{-\theta})$.

The proof completes by the fact that, if $f(x)=o(h(x))$ and $g(x)=\Theta(h(x))$ then $f(x)+g(x)=\Theta(h(x))$.

\section{Proof of Theorem \ref{theorem-fat-product-sufficient}}
\label{proof-theorem-fat-product-sufficient}

We prove the case of $N=2$ and the proof of the case $N>2$ follows by the iteration of the same procedure.

Considering the independence between $X_1$ and $X_2$, for $x>0$, we have $\overline{F}_{X_1 X_2}(x) = \mathbb{E}\qty[ \overline{F}_{X_1} \qty(\frac{x}{X_2}) ]$, which is reformulated as
\begin{multline}
\mathbb{E}\qty[ \overline{F}_{X_1} \qty(\frac{x}{X_2}) ] = \mathbb{E}\qty[ \overline{F}_{X_1} \qty(\frac{x}{X_2}) 1_{0< X_2 \le \varphi_2(x) } ] \\
+ \mathbb{E}\qty[ \overline{F}_{X_1} \qty(\frac{x}{X_2}) 1_{ X_2 > \varphi_2(x) } ].
\end{multline}

Since $0\le \overline{F}_{X_1}\qty(x) \le 1$, we have $0 \le \mathbb{E}\qty[ \overline{F}_{X_1} \qty(\frac{x}{X_2}) 1_{ X_2 > \varphi_2(x) } ] \le \mathbb{P} \qty( X_2 > \varphi_2(x) ) = \overline{F}_{X_2}\qty( \varphi_2(x) )$. 
Thus, $\mathbb{E}\qty[ \overline{F}_{X_1} \qty(\frac{x}{X_2}) 1_{ X_2 > \varphi_2(x) } ] = O\qty(x^{-\theta})$.

Since $\overline{F}_{X_1}(x) = O\qty(x^{-\theta})$ and $\lim_{x\rightarrow\infty}\frac{x}{X_2} \ge \lim_{x\rightarrow\infty} \frac{x}{\varphi_2(x)} = \infty$, we have,
$\exists x_0 >0$, $\forall x > x_0$, 
$\exists C_1>0$, $\exists \theta>0$, 
$\mathbb{E}\qty[ \overline{F}_{X_1} \qty(\frac{x}{X_2}) 1_{0< X_2 \le
\varphi_2(x) } ] \le \mathbb{E}\qty[ C_1 \qty(\frac{x}{X_2})^{-\theta}  ]
=  C_1 \mathbb{E}\qty[ X_2^{\theta} ] x^{-\theta}$. Thus,
$\mathbb{E}\qty[ \overline{F}_{X_1} \qty(\frac{x}{X_2}) 1_{0< X_2 \le \varphi_2(x) } ] = O\qty(x^{-\theta})$.

Considering the $O(\cdot)$ polynomial \cite{hein2015discrete}, i.e., if $f_1(x)=O(g(x))$ and $f_2(x)=O(g(x))$ then $ f_1(x) + f_2(x) = O(g(x)) $,
we have $\overline{F}_{X_1 X_2}(x) = \mathbb{E}\qty[ \overline{F}_{X_1} \qty(\frac{x}{X_2}) ] = O\qty(x^{-\theta})$.

\section{Proof of Theorem \ref{theorem-super-heavy-tail-bound}}
\label{proof-theorem-super-heavy-tail-bound}

We prove the case of $N=2$ and the proof of the case $N>2$ follows by the iteration of the same procedure and by the fact that \cite{hein2015discrete}, if $f(x)= O(g(x))$ and $g(x)=O(h(x))$ then $f(x)=O(h(x))$. 

Considering the independence between $X_1$ and $X_2$, for $x>0$, we have $\overline{F}_{X_1 X_2}(x) = \mathbb{E}\qty[ \overline{F}_{X_1} \qty(\frac{x}{X_2}) ]$, which is reformulated as
\begin{multline}
\mathbb{E}\qty[ \overline{F}_{X_1} \qty(\frac{x}{X_2}) ] = \mathbb{E}\qty[ \overline{F}_{X_1} \qty(\frac{x}{X_2}) 1_{0< X_2 \le \varphi_2(x) } ] \\
+ \mathbb{E}\qty[ \overline{F}_{X_1} \qty(\frac{x}{X_2}) 1_{ X_2 > \varphi_2(x) } ].
\end{multline}

Since $0\le \overline{F}_{X_1}\qty(x) \le 1$, we have $0 \le \mathbb{E}\qty[ \overline{F}_{X_1} \qty(\frac{x}{X_2}) 1_{ X_2 > \varphi_2(x) } ] \le \mathbb{P} \qty( X_2 > \varphi_2(x) ) = \overline{F}_{X_2}\qty( \varphi_2(x) )$. 
Thus, $\mathbb{E}\qty[ \overline{F}_{X_1} \qty(\frac{x}{X_2}) 1_{ X_2 > \varphi_2(x) } ] = O\qty( \overline{F}_{X_1}(x) )$.

Since $\overline{F}_{X_1}(x) $ is nonincreasing, 
we have $\mathbb{E}\qty[ \overline{F}_{X_1} \qty(\frac{x}{X_2}) 1_{0< X_2 \le
\varphi_2(x) } ] \le \overline{F}_{X_1} \qty(\frac{x}{\varphi_2(x)}) $. 
Thus, $\mathbb{E}\qty[ \overline{F}_{X_1} \qty(\frac{x}{X_2}) 1_{0< X_2 \le
\varphi_2(x) } ] = O\qty( \overline{F}_{X_1}(x) ) $.

Considering the $O(\cdot)$ polynomial \cite{hein2015discrete}, i.e., if $f_1(x)=O(g(x))$ and $f_2(x)=O(g(x))$ then $ f_1(x) + f_2(x) = O(g(x)) $,
we have $\overline{F}_{X_1 X_2}(x) = \mathbb{E}\qty[ \overline{F}_{X_1} \qty(\frac{x}{X_2}) ] = O\qty(x^{-\theta})$.

The proof of the result for the sum of random variables follows analogically.

\section{Proof of Theorem \ref{theorem-slowly-regularly-varying-dominate}}
\label{proof-theorem-slowly-regularly-varying-dominate}

We prove the case of $N=2$ and the proof of the case $N>2$ follows by the iteration of the same procedure. 

Considering the independence between $X_1$ and $X_2$, for $x>0$, we have $\overline{F}_{X_1 X_2}(x) = \mathbb{E}\qty[ \overline{F}_{X_1} \qty(\frac{x}{X_2}) ]$, which is reformulated as
\begin{multline}
\mathbb{E}\qty[ \overline{F}_{X_1} \qty(\frac{x}{X_2}) ] = \mathbb{E}\qty[ \overline{F}_{X_1} \qty(\frac{x}{X_2}) 1_{0< X_2 \le \varphi_2(x) } ] \\
+ \mathbb{E}\qty[ \overline{F}_{X_1} \qty(\frac{x}{X_2}) 1_{ X_2 > \varphi_2(x) } ].
\end{multline}

Since $0\le \overline{F}_{X_1}\qty(x) \le 1$, we have $0 \le \mathbb{E}\qty[ \overline{F}_{X_1} \qty(\frac{x}{X_2}) 1_{ X_2 > \varphi_2(x) } ] \le \mathbb{P} \qty( X_2 > \varphi_2(x) ) = \overline{F}_{X_2}\qty( \varphi_2(x) )$. 
Thus, $\mathbb{E}\qty[ \overline{F}_{X_1} \qty(\frac{x}{X_2}) 1_{ X_2 > \varphi_2(x) } ] = o\qty( \overline{F}_{X_1}(x) )$.

In addition, 
we have $\mathbb{E}\qty[ \overline{F}_{X_1} \qty(\frac{x}{X_2}) 1_{0< X_2 \le
\varphi_2(x) } ] \sim \mathbb{E} \qty[ \overline{F}_{X_1} \qty( \frac{x}{X_2} ) ] $, which follows the Lebesgue dominated convergence theorem. 
Since $ \mathbb{E} \qty[ \overline{F}_{X_1} \qty( \frac{x}{X_2} ) ] = \int \overline{F}_{X_1} \qty( \frac{x}{x_2} ) d F_{X_2}(x_2) \sim \int \overline{F}_{X_1}(x) d F_{X_2}(x_2) = \overline{F}_{X_1}(x)$, which follows that $\overline{F}_{X_1}$ is slowly varying, i.e., $\overline{F}_{X_1} \qty(tx) \sim \overline{F}_{X_1} \qty(x)$, $\forall t>0$,
thus, we obtain $\mathbb{E}\qty[ \overline{F}_{X_1} \qty(\frac{x}{X_2}) 1_{0< X_2 \le
\varphi_2(x) } ] \sim \overline{F}_{X_1}(x) $.

Considering the asymptotics polynomial \cite{hein2015discrete}, i.e., if $f_1(x)=o(g(x))$ and $f_2(x) \sim g(x)$ then $ f_1(x) + f_2(x) \sim g(x) $,
we have $\overline{F}_{X_1 X_2}(x) = \mathbb{E}\qty[ \overline{F}_{X_1} \qty(\frac{x}{X_2}) ] \sim \overline{F}_{X_1}(x)$.

The proof of the result for the sum of random variables follows analogically.

\balance
\bibliographystyle{IEEEtran}
\bibliography{main}

\end{document}